\documentclass[a4paper,USenglish,cleveref, autoref, thm-restate,numberwithinsect]{lipics-v2021}
\nolinenumbers
\hideLIPIcs
\usepackage{apxproof}

\usepackage{todonotes}

\newcommand{\var}{\mathrm{var}}
\newcommand{\atoms}{\mathrm{atoms}}
\newcommand{\arity}{\mathrm{arity}}
\newcommand{\dom}{\mathrm{dom}}
\newcommand{\eps}{\varepsilon}
\newcommand{\qstar}{Q^\star}
\newcommand{\col}{\mathcal{L}}
\newcommand{\dep}{\mathop{dep}}

\newcommand{\polylog}{\mathrm{polylog}}
\newcommand{\cimply}{\rightarrow^*}

\newcommand{\pb}[2]{P_{#1,#2}}
\newcommand{\cn}[2]{C_{#1,#2}}
\newcommand{\fec}[1]{\rho^*_{#1}}
\newcommand{\width}[1]{#1\text{-width}}

\makeatletter
\def\namedlabel#1#2{\begingroup
    #2%
    \def\@currentlabel{#2}%
    \phantomsection\label{#1}\endgroup
}

\newtheoremrep{theorem}{Theorem}[section]
\newtheoremrep{corollary}[theorem]{Corollary}
\newtheoremrep{proposition}[theorem]{Proposition}
\newtheoremrep{definition}[theorem]{Definition}
\newtheoremrep{lemma}[theorem]{Lemma}
\newtheoremrep{observation}[theorem]{Observation}

\title{Lexicographic Direct Access with Functional Dependencies}

\author{Florent Capelli}{Univ.~Artois, CNRS, Centre de Recherche en Informatique de Lens (CRIL), Lens, France}{}{}{}

\author{Nofar Carmeli}{Inria, LIRMM, University of Montpellier, CNRS, France}{Nofar.Carmeli@inria.fr}{https://orcid.org/0000-0003-0673-5510}{}

\author{Stefan Mengel}{Univ.~Artois, CNRS, Centre de Recherche en Informatique de Lens (CRIL), Lens, France}{}{}
{}

\Copyright{F.~Capelli, N.~Carmeli, S.~Mengel}
\authorrunning{F.~Capelli, N.~Carmeli, S.~Mengel}
\acknowledgements{The authors would like to thank Mahmoud Abo Khamis for helpful clarifications on the PANDA algorithm.}

\begin{document}
\maketitle

\begin{abstract}
We study the complexity of lexicographic direct access to join query answers over databases that satisfy functional dependencies (FDs). More precisely, we give fine-grained lower and upper bounds on the preprocessing time required to achieve polylogarithmic access time.
We start by considering the simple approach of a reordered extension, which first incorporates the FDs in the query and order, and then ignores the FDs during evaluation.
We show that this simple approach gives tight bounds for unary FDs but fails for general FDs.
We then consider a second approach, inspired by size bounds for query answers using information theory, that takes the FDs into account while materializing the bags of a decomposition tailored to the direct access task at hand. Interestingly, we show that the same reordering is also useful while constructing the decomposition in this second approach for reducing the complexity.
While the obtained upper and lower bounds are generally not tight, we show that they yield a complete characterization of lexicographic direct access with linear preprocessing time.
All lower bounds in this paper apply only to queries without self-joins and rely on the Zero-Clique Conjecture.
\end{abstract}


\section{Introduction}

Recent years have seen steady progress in the understanding of the complexity of conjunctive query evaluation. On the one hand, worst-case optimal evaluation algorithms~\cite{NgoPRR18,Veldhuizen14,Ngo2018,CapelliIS25} have been a breakthrough and have recently been combined with fast matrix multiplication to give even better runtimes~\cite{DeepHK20,KhamisHS25,Hu25}. On the other hand, techniques and conjectures from fine-grained complexity theory have been used to show lower bounds for many different query answering settings. In particular, this has lead to a good understanding of the classes of queries that can be solved in linear time for many query answering settings~\cite{BerkholzGS20,Mengel25}.

However, almost all of the above lower bound results are only for the case without any constraints in the input database, while on the algorithmic side it is known that different constraints on the data can lead to huge speedups~\cite{KhamisNS25,DeedsM25}.
Functional Dependencies (FDs)~\cite[Part~C]{AbiteboulHV95} form a simple class of database constraints of great practical importance, especially since they are crucial for database normalization~\cite{Kent83}. 
This combination of simplicity and importance has led to a significant amount of work on them in database theory~\cite{GottlobLVV12,GogaczT17,Khamis0S16,carmeli2020enumeration,abo2016computing}, and yet, even for FDs, very few lower bounds are known~\cite{carmeli2020enumeration, carmeli2023linearp}, and linear-time characterizations are largely lacking.

In this paper, we make a step towards a better understanding of the complexity of join query answering under constraints. 
Even though we focus on join queries, some of our results also apply to conjunctive queries (including projection).
We study lexicographic direct access to query results under FDs. The aim of direct access is to, after some preprocessing, allow access to the query answers as if they were materialized into a sorted array.
Often it is possible to do so without actually computing the query result, using a preprocessing that is far more efficient~\cite{BaganDGO08,bb:thesis,carmeli2023linearp,bringmann2024stars}.
We remark that direct access can be used to efficiently compute quantiles and histograms and also to support sampling without repetitions~\cite{carmeli2022random}.
We focus on direct access in which the underlying query order is lexicographical, according to a variable order defined by the user together with the query.
Without FDs, we know the optimal preprocessing required to get access in logarithmic time~\cite{bringmann2024stars}.
The hardness part of this result, like all hardness results in the RAM computation model, is a conditional lower bound, relying on the assumed hardness of a well-studied problem. In this case, just like the lower bounds we show in this paper, it assumes the Zero-Clique Conjecture.
We here study how this complexity changes in the presence of FDs.

For lexicographic direct access over databases satisfying a given set of FDs, Tziavelis et al.~\cite{carmeli2023linearp} studied the case that all FDs are unary (i.e., one variable implies another), and showed a dichotomy that identifies which combinations of conjunctive query, variable order, and FD set admit an algorithm with quasilinear preprocessing and logarithmic access time.
The hardness part of this dichotomy, like most lower bounds in the field and the lower bounds shown in this paper, is restricted to queries without self-joins; that is, every query atom uses a distinct relation.
This dichotomy is obtained with the approach of reordered extensions.
The idea of this approach is to rewrite the query to incorporate the extra information that can be obtained from the FDs in order to end up in an equivalent setting without FDs. Then, known results for the general setting can be applied to derive the complexity of the setting with FDs. For this to work also for the lower bounds, the specified variable order is reordered to take the FDs into account. We show that for unary FDs, this simple approach is surprisingly powerful, beyond the linear preprocessing dichotomy: it allows to completely understand the complexity of lexicographic direct access under unary FDs, identifying the optimal preprocessing time required to obtain logarithmic access time.
We generalize the extension-based approach beyond unary FDs. However, we show that, in this general case, the obtained algorithm is not optimal. Interestingly, this gap occurs even in the restricted case that the underlying query is acyclic, in contrast to the related problem of enumeration where the extension-based approach is known to be tight for acyclic queries~\cite{carmeli2020enumeration}.

Following the failure of reordered extensions for direct access, we turn to techniques based on information theory, and in particular, the celebrated PANDA algorithm for join evaluation~\cite{KhamisNS25}.
We combine it with direct access techniques~\cite{bringmann2024stars,carmeli2023linearp} by taking the FDs into account when inspecting the complexity of materializing the bags of a so-called disruption-free-decomposition. 
Interestingly, we show that the same reordering from the extension approach is also useful while constructing the decomposition in this second approach for reducing the complexity.
To complement the algorithm with lower bounds, we use a coloring technique that was previously used to bound the size of query results~\cite{GottlobLVV12}. Unfortunately, our complexity upper and lower bounds are generally not tight, which stems from the fact that the corresponding bounds on join sizes based on the so-called polymatroid bound are generally also not tight. Thus, similarly to the case of join size computation and worst-case optimal join algorithms, fully resolving the complexity of direct access under FDs would likely require fundamental advances in information theory.
We do show that our algorithms and lower bounds form a dichotomy characterizing the combinations of join query, variable order, and FD set that admit lexicographic direct access with linear preprocessing and logarithmic access time (for general FDs). 

Finally, we compare the two approaches presented in this paper. We show that the algorithm based on information-theory techniques is always at least as good as the extension-based algorithm and sometimes better, but that in the case of unary FDs, both algorithms yield the same complexity.

We give preliminary notation and recall useful known results in \Cref{sec:prelim}.
We study the extension-based approach in \Cref{sec:reorder}.
For our second approach, the algorithm is devised in \Cref{sct:algorithm}, the lower bound in \Cref{sec:lower}, and the linear preprocessing dichotomy in \Cref{sec:sizepreserving}.
Finally, we compare the two approaches in \Cref{sec:connection}, and we conclude in \Cref{sec:conclusion}.

\section{Preliminaries}\label{sec:prelim}

Given a natural number $n$, we set $[n]:=\{1,\ldots,n\}$.
Given a tuple $\vec{a}=(a_1,\ldots,a_n)$, we set $\vec{a}[i]:=a_i$.
Given an ordering $\pi$, we write $a\le_\pi b$ to say that $a$ appears before $b$ in $\pi$.

\textbf{Databases.}
A \emph{schema} is a set of relational symbols.
Each relational symbol $R$ is associated with a natural number $\arity(R)$, called the arity of $R$.
A \emph{database} $D$ over the schema $S$ 
contains a finite relation $R^D \subseteq \dom^{\arity(R)}$ for each $R \in S$, 
where $\dom$ is a set of constant values called the \emph{domain}.
We sometimes write $R$ instead of $R^D$ when $D$ is clear from the context.
If $(c_1,\dots,c_{\arity(R)})\in R^D$, then the expression $R(c_1,\dots,c_{\arity(R)})$ is called a \emph{fact} of $D$.
The size of a database $D$, denoted $|D|$, is the number of facts it contains.
The database schemas we consider are extended with Functional Dependencies (FDs). An FD is an expression of the form $R: A \rightarrow B$, where $A,B \subseteq [\arity(R)]$. A database \emph{satisfies} such an FD if for every two facts $R(t_1)$, $R(t_2)$ of $D$, we have that if $t_1[i] = t_2[i]$ for all $i \in \mathbf{A}$, then 
$t_1[j] = t_2[j]$ for all $j \in \mathbf{B}$.
If the schema $S$ contains a set $\Delta$ of FDs, a database over $S$ must satisfy all FDs in $\Delta$. 
We may assume without loss of generality that all FDs are of the form $R_i: A \rightarrow b$ where $b$ is a single variable because we can replace an FD of the form $R_i: A \rightarrow B$
with a set of FDs $\{ R_i: A \rightarrow b \mid b \in B \}$.
If $|A| = 1$, the FD is called \emph{unary}.

\textbf{Queries.}
A \emph{conjunctive query} (CQ) $Q$ over schema $S$ is an expression of the form
$Q(\vec x) \leftarrow R_1(\vec x_1), \ldots, R_\ell(\vec x_\ell)$,
where the tuples $\vec x, \vec x_1, \ldots, \vec x_\ell$ hold variables, 
every variable in $\vec x$ appears in some $\vec x_1, \ldots, \vec x_\ell$, we have that $R_1, \ldots, R_\ell \in S$, and $\vec x_i$ has the same arity as $R_i$ for all $i\in [\ell]$.
Each $R_i(\vec x_i)$ is called an \emph{atom} of the query $Q$, 
and the set of all atoms is denoted by $\atoms(Q)$.
When we do not need to refer to the order of variables within an atom $R_i(\vec x_i)$, we sometimes denote it by $R_i(X_i)$ for simplicity, where $X_i$ is a set of variables.
We use $\var(Q)$ for the set of variables that appear in $Q$.
The variables $\vec x$ are called \emph{free}.
We say that variables that are not free are \emph{projected}.
If the CQ does not contain projections, i.e., all query variables are free, it is called a \emph{join query}.
A repeated occurrence of a relational symbol in two different atoms is a \emph{self-join}, and if no self-joins exist, the CQ is called \emph{self-join-free}.
A homomorphism from a CQ $Q$ to a database $D$ is a mapping of $\var(Q)$ to constants from $\dom$ such that every atom of $Q$ maps to a fact of $D$.
A \emph{query answer} is such a homomorphism 
followed by a projection on the free variables. The set of query answers is denoted $Q(D)$.
In the context of answering a CQ, we can assume without loss of generality that FDs are expressed using the query variables~\cite[Section 8]{carmeli2023linearp}.
That is, if the database contains the FD $R_i: A \rightarrow b$, and the query contains the atom $R_i(\vec x_i)$, we can denote by $Y$ and $\{z\}$ the sets of variables in positions $A$ and $\{b\}$ respectively of $\vec x_i$, and specify the FD as $R_i: Y \rightarrow z$.
When $R_i$ is clear or irrelevant, we denote the FD simply by $Y \rightarrow z$, and we say that $Y$ implies $z$.
We write $Y\cimply Z$ when every variable of $Z$ is transitively implied by some subset of $Y$.

\textbf{Hypergraphs and query structure.}
A {\em hypergraph} $(V,E)$ is a set $V$ of {\em vertices} and a set $E$ of subsets of $V$ called {\em hyperedges}.
Two vertices are {\em neighbors} if they appear in the same hyperedge.
A {\em path} is a sequence of vertices such that every two succeeding vertices are neighbors.
Given a set $S$ of vertices, we denote by $N(S)$ the set containing all vertices that have a neighbor in $S$.
Two vertices are \emph{connected} if there is a path between them. A \emph{connected component} is a maximal set of vertices such that every two vertices of the set are connected.
A \emph{join tree} of a hypergraph $(V,E)$ is a tree where the nodes are the hyperedges $E$ and the {\em running intersection} property holds, namely: 
for all $u \in V$ the set $\{e \in E \mid u \in e\}$ forms a (connected) subtree.
A hypergraph is {\em acyclic} if it has a join tree.
We associate a hypergraph $H_Q = (V, E)$ to a CQ $Q$ 
where the vertices are the variables of $Q$, 
and every atom of $Q$ corresponds to a hyperedge with the same set of variables.
A CQ $Q$ is {\em acyclic} if $H_Q$ is acyclic,
and otherwise it is \emph{cyclic}.

\textbf{Lexicographic direct access.}
\emph{Lexicographic direct access}, sometimes simply called direct access, is a task defined by a query $Q$ and an ordering $\pi$ of its free variables. Given an input database $D$, after a preprocessing phase, the user can specify an index $j$ and expects the $j$th answer in $Q(D)$ according to the lexicographic order corresponding to $\pi$ or an out-of-bounds error if there are fewer than $j$ answers. When the valid inputs to the problem are restricted to those satisfying a set $\Delta$ of FDs, we specify such a task by the triple $(Q,\pi,\Delta)$. The time it takes to provide an answer given an index is called the \emph{access time}.
As in previous work, as a computation model, for an input of size $n$, we use the word-RAM model with $O(\log(n))$-bit words and unit-cost operations.

Given $k\in \mathbb{N}$, the $k$-star query is defined as $\qstar_k(x_1, \ldots,x_{k+1}) \leftarrow \bigwedge_{i\in [k]} R_i^\star(x_i, x_{k+1})$. 
A variable ordering for $\qstar_k$ is called \emph{bad} if $x_{k+1}$ is last.
The following hardness result is obtained by combining Proposition 19, Lemma 22, and Theorem 23 by Bringmann et al.~\cite{bringmann2024stars}.

\begin{theorem}[\cite{bringmann2024stars}]\label{thm:star-hard}
    Let $2\le k\in\mathbb{N}$. Assuming the Zero-Clique Conjecture, there is no $\eps>0$ such that for all $\delta>0$ there is a direct access algorithm for $\qstar_k$ with respect to a bad ordering with preprocessing time $O(|D^\star|^{k-\varepsilon})$ and access time $O(|D^\star|^{\delta})$.
\end{theorem}

In the \emph{Zero-$k$-Clique} problem, given an $n$-node graph with integer edge weights in $[-n^c,\ldots, n^c]$ for some constant $c\ge 1$, the task is to decide whether the graph contains a $k$-clique with total weight $0$. 
The \emph{Zero-Clique Conjecture} states that, for every $k \ge 3$, there is no constant $\varepsilon >0$ such that \emph{Zero-$k$-Clique} has a randomized algorithm running in time $O(n^{k-\varepsilon})$~\cite{bringmann2024stars}.

\textbf{Disruption-free decompositions and the incompatibility number.}
Let $Q(X)\leftarrow R_1(X_1), \ldots, R_\ell(X_\ell)$ be a join query and $\pi=(v_1,\ldots,v_n)$ an
ordering of its variables.
For every $i\in [n]$, we denote by $Q[v_i, \ldots, v_n]$ the query $Q(X\cap \{x_i, \ldots, x_n\}) \leftarrow R_1(X_1\cap \{x_i, \ldots, x_n\}), \ldots, R_\ell(X_\ell\cap \{x_i, \ldots, x_n\})$.
Denote by $S_i$ the vertices in the connected component of $v_i$ in $Q[{v_i,\ldots,v_n}]$.
The \emph{disruption-free decomposition} of $Q$ according to $\pi$ is defined\footnote{In the definition by Bringmann et al.~\cite{bringmann2024stars}, disruption-free decompositions are also defined to contain a hyperedge for every atom of $Q$. Since every atom is contained in the bag created by its largest variable, and adding or removing hyperedges contained in the created bags does not make a difference for our purposes, we chose the current definition for simplicity.} to be the hypergraph $(\var(Q),\{B_1,\ldots,B_n\})$ where $B_i:=\{v_i\}\cup\{v_j\mid j<i , v_j\in N(S_i)\}$ for $i\in[n]$. We call the sets $B_i$ the \emph{bags} of the decomposition and say that $v_i$ is the variable \emph{creating} the bag $B_i$. 
The join query with an atom $R_i(B_i)$ for every bag $B_i$ of the disruption-free decomposition admits lexicographic direct access according to $\pi$ with $O(|D|)$ preprocessing and $O(\log(|D|))$ access time over a database $D$~\cite{bringmann2024stars,carmeli2023linearp}.

Given a hypergraph $H=(V,E)$ and a set of vertices $S\subseteq V$, a fractional edge cover of $S$ using $H$ is a mapping $\mu: E\rightarrow [0,1]$ such that for every $v\in S$ we have $\sum_{e: v\in e} \mu(e) \ge 1$.
The \emph{fractional edge cover number} of $S$ using $H$ is $\fec{H}(S) := \min_\mu {\sum_{e\in E} \mu(e)}$, where the minimum is taken over all fractional edge covers of $S$ using $H$. We remark that it can be computed efficiently by linear programming.
The \emph{incompatibility number} $\iota$ of $Q$ and $\pi$ is defined to be the maximum fractional edge cover number of a bag of the disruption-free decomposition using the query hypergraph. That is,
$\iota(Q,\pi)=\max_{i\in [n]}{\fec{Q}(B_i)}$.

\begin{example}
Consider $Q(v_1, v_2, v_3, v_4, v_5) \leftarrow R_1(v_1, v_4), R_2(v_4,v_3), R_3(v_3, v_5), R_4(v_5, v_2)$ and $\pi= (v_1,v_2,v_3,v_4,v_5)$. We get that $S_5 = \{v_5\}$, $S_4 = \{v_4\}$, $S_3 = \{v_3, v_4, v_5\}$, $S_2 = \{v_2, v_3, v_4,v_5\}$, and $S_1= \{v_1, v_2, v_3, v_4, v_5\}$. The resulting bags are $B_5 = \{v_2, v_3, v_5\}$, $B_4 = \{v_1, v_3, v_4\}$, $B_3= \{v_1, v_2, v_3\}$, $B_2 = \{v_1, v_2\}$, and $B_1 = \{v_1\}$. The biggest fractional edge cover number out of these bags is that of $B_3$, which is $3$, so $\iota(Q,\pi)=3$.
\end{example}

The incompatibility number determines the preprocessing time for direct access as follows:
\begin{theorem}[\cite{bringmann2024stars}]\label{thm:bringmann-et-al}
    Given a join query $Q$ and an ordering $\pi$ of its variables, 
		\begin{itemize}
		  \item lexicographic direct access for $(Q,\pi)$ is possible with $O(|D|^{\iota(Q, \pi)})$ preprocessing time and $O(\log(|D|))$ access time, and
		  \item there is no constant $\varepsilon > 0$ such that for all $\delta > 0$ Lexicographic direct access for $(Q,\pi)$ is possible with preprocessing time $O(|D|^{\iota(Q, \pi)-\varepsilon})$ and access time $O(|D|^\delta)$, assuming the Zero-Clique Conjecture.
		\end{itemize}
\end{theorem}

\textbf{The Polymatroid Bound and the PANDA Algorithm.}
We next present some basics on the polymatroid bound, slightly adapted to our setting.

Let $V$ be a finite set. A \emph{set function} is then a function $h:2^V \rightarrow \mathbb R_+$. If for all $A\subseteq B\subseteq V$ we have $h(A)\le h(B)$, we say that $h$ is \emph{monotone}. We call $h$ \emph{submodular} if for all $A,B\subseteq V$ we have $h(A) + h(B) \ge h(A\cup B) + h(A\cap B)$. If $h$ is monotone and submodular and we have $h(\emptyset)=0$, we also call it a \emph{polymatroid}.
We will only be interested in polymatroids whose ground set $V$ is the set of variables of a CQ. So fix a query $Q(X):= \bigwedge_{i\in [k]} R_i(X_i)$ and let $V:= \bigcup_{i\in [k]} X_i$. We say that a polymatroid $h$ is guarded by $Q$ if for all $i\in [k]$, we have $h(X_i) \le 1$.
We say that $h$ respects an FD $Y\rightarrow z$, if $h(z\mid Y) := h(Y\cup \{z\}) - h(Y) = 0$.

We will be interested in the set function defined for every $S\subseteq V$ by $\pb{Q}{\Delta}(S) := \max_h h(S)$ where the maximum is taken over all polymatroids that are guarded by $Q$ and respect all FDs in $\Delta$. 
Remark that $\pb{Q}{\Delta}(S)$ is the solution of a linear program with integer coefficients with the variable set $\{h(Y)\mid Y\subseteq V\}$ and as constraints the submodularity and monotonicity requirements as well as the guardedness. Since we assume that every variable of $Q$ is in an atom, $\pb{Q}{\Delta}(S)$ is always bounded by a finite number; for example, $|S|$ is a very crude upper bound because $h(S)\le \sum_{x\in S} h(\{x\})\le |S|$ holds for every $h$ by submodularity and monotonicity and the fact that for every $x$ there is $i\in [k]$ with $x\in X_i$ and thus $h(\{x\}) \le h(X_i) \le 1$. As the solution of a linear program with a bounded target function, $\pb{Q}{\Delta}(S)$ is computable and always a rational number.
We are interested in $\pb{Q}{\Delta}(S)$ because of the following result often called the polymatroid bound\footnote{To avoid confusion, we remark that in~\cite{Khamis0S16}, $\pb{Q}{\Delta}(S)$, using different notation, is scaled by a logarithmic factor compared to~\cite{GottlobLVV12}. This difference is purely arithmetic and does not change the \emph{polymatroid bound}. We here use the variant from~\cite{GottlobLVV12} since it better matches our asymptotic perspective.}.
\begin{theorem}[\cite{GottlobLVV12,Khamis0S16}]\label{thm:polymatroid}
    Let $Q(X)$ be a CQ and $\Delta$ a set of FDs. For every database $D$ that respects the constraints in $\Delta$, we have that $|Q(D)| \le |D|^{\pb{Q}{\Delta}(X)}$.
\end{theorem}

Due to \Cref{thm:polymatroid}, we call the function~$\pb{Q}{\Delta}$ the \emph{exponent of the polymatroid bound}.

A fundamental result from~\cite{Khamis0S16} is that the polymatroid bound is essentially also a runtime bound for CQ evaluation. In particular, there is an algorithm (called PANDA) that for every CQ $Q(X)$ with a set of FDs $\Delta$ does the following: given a database $D$ respecting $\Delta$, PANDA computes a relation $R \supseteq Q(D)$ in time $\polylog(|D|) |D|^{\pb{Q}{\Delta}(X)}$.

Note that, we cannot expect the PANDA algorithm to compute $Q(D)$ exactly, since doing so might be computationally hard even if the polymatroid bound is small. For example, $Q$ can be a clique query of which we only want the projection to a single variable. This query has polymatroid bound $1$, but for big enough cliques we cannot even expect to decide in linear time if it has any answers, assuming the exponential time hypothesis, see e.g.~\cite{LokshtanovMS11}.
When considering join queries (without projections), this does not happen: in that case, we can always simply filter the relation $R$ computed by PANDA with all atoms to compute $Q(D)$ in time $\polylog(|D|) |D|^{\pb{Q}{\Delta}(X)}$, see~\cite{Khamis0S16} for details.

\section{Reordered Extensions}\label{sec:reorder}

In this section, we inspect the approach of achieving direct access with FDs using a reduction to a reordered extension without dependencies, and we show that this approach yields optimal preprocessing time (up to sub-polynomial factors) for logarithmic access time in the case of unary FDs.
We first provide a natural generalization of $\Delta$-reordered extensions~\cite{carmeli2023linearp} for general FDs.

\begin{definition}[$\Delta$-reordering]\label{def:FD-reordering}
    Given a CQ $Q$, a set of FDs $\Delta$, and an ordering $\pi$ of variables from $Q$, we say that $\pi'$ is a \emph{$\Delta$-reordering} of $\pi$ if it can be obtained by the following procedure: Initialize $\pi'$ empty. Then, for every variable $v$ according to the order $\pi$: if $v$ is not in $\pi'$, append it to $\pi'$; then, for every FD $Y \rightarrow y$ with $Y$ in $\pi'$ and $y$ not in $\pi'$, append $y$ to $\pi'$.
    We say that $\pi$ is \emph{consistent} with respect to $\Delta$ if it is a $\Delta$-reordering of itself.
\end{definition}

We remark that a problem may admit more than one $\Delta$-reordering, and any choice of reordering will work for our purposes. As an example, given the FDs $x\rightarrow y$ and $x\rightarrow z$ and $\pi=(x)$, both $(x,y,z)$ and $(x,z,y)$ are valid reorderings. 
The reordering serves as a way of representing the problem's complexity more directly without changing the query answering problem. 
In other words, for every database $D$ that satisfies $\Delta$, ordering $Q(D)$ by $\pi$ is the same as ordering it by its $\Delta$-reordering $\pi'$.

\begin{observationrep}\label{lem:reordering-validity}
Given a CQ $Q$, a set of FDs $\Delta$, a variable order $\pi$, and a $\Delta$-reordering $\pi'$ of $\pi$, then the triples 
$(Q,\pi',\Delta)$ and $(Q,\pi,\Delta)$ induce the same output in the same order.
\end{observationrep}
\begin{appendixproof}
Given an FD $Y\rightarrow z$, once the assignment for $Y$ is set, $z$ can have at most one possible assignment, so as long as $z$ comes after all variables of $Y$ in the order, its exact position or whether it is considered as part of the order does not influence the answer ordering. This observation was already stated for unary FDs~\cite[Lemma 8.16]{carmeli2023linearp}.
\end{appendixproof}

$\Delta$-reordered extensions extend the atoms and head of CQs using the FDs, in addition to reordering the required variable order.

\begin{definition}[$\Delta$-reordered extension]\label{def:reordered-ext}
    Given a self-join free CQ $Q(X)$, a set of FDs $\Delta$, and a variable order $\pi$, their $\Delta$-reordered extension is a CQ $Q^+(X^+)$, a set of FDs $\Delta^+$ and a variable order $\pi^+$ defined as follows.
    \begin{itemize}
        \item $Q^+$ body and $\Delta^+$ are obtained as the fixpoint of the following extension step: Given an FD $Y \rightarrow z\in\Delta$, if $Y \subseteq X_i$ and $z \notin X_i$ for some atom $R_i(X_i)$, then increase the arity of $R_i$ by one, replace $R_i(X_i)$ with $R_i(X_i, z)$, and add $R_i:Y \rightarrow z$ to the FD set $\Delta^+$.
        \item $Q^+$ head is obtained as the fixpoint of the following extension step, starting with $X^+ := X$: given an FD $Y \rightarrow z\in\Delta$, if $Y \subseteq X^+$ and $z \notin X^+$, then add $z$ to $X^+$.
        \item $\pi^+$ is a $\Delta$-reordering of $\pi$.
    \end{itemize}
\end{definition}

We will next use exact reductions to connect the complexity of a direct access problem and its $\Delta$-reordered extension.
Given a variable order $\pi$, and query answers $a_1, a_2$, we denote by $a_1 \preceq_{\pi} a_2$ the fact that $a_1$ appears before $a_2$ in the lexicographic order defined by $\pi$.

Consider CQs $Q$ and $Q'$, FD sets $\Delta$ and $\Delta'$, and variable orders $\pi$ and $\pi'$.
We say that there is an \emph{exact reduction} from $(Q,\pi,\Delta)$ to $(Q',\pi',\Delta')$ if for every database $D$ that satisfies $\Delta$, we can construct in linear time a database $D'$ that satisfies $\Delta'$ such that there is a bijection $\tau$ from $Q'(D')$ to $Q(D)$ that is computable in constant time, and for all query answers $a_1, a_2$ of $Q'$, we have that $a_1 \preceq_{\pi'} a_2$ iff $\tau(a_1) \preceq_\pi \tau(a_2)$.

\begin{toappendix}
    
As any database that satisfies a set of FDs $\Delta$ also satisfies $\emptyset$, using the identity for the construction and bijection yields a trivial exact reduction from a problem with FDs to the same problem without FDs. 

\begin{observation}\label{obs:remore-fds}
    Given a self-join-free CQ $Q$, a variable order $\pi$, and a set of FDs $\Delta$, there is an exact reduction from $(Q,\pi,\Delta)$ to $(Q,\pi,\emptyset)$.
\end{observation}

Carmeli and Kr\"oll~\cite[Theorem 2]{carmeli2020enumeration} show exact reductions in both directions between a query and its reordered extension. Inspecting their proofs, we notice that they preserve the answer order.

\begin{lemma}\label{lem:lex-reduct}\label{lem:ext-to-orig}
Given a self-join-free CQ $Q$, a variable order $\pi$, and a set of FDs $\Delta$, there are exact reductions in both directions between $(Q,\pi,\Delta)$ and $(Q^+,\pi,\Delta^+)$.
\end{lemma}
\begin{proof}
    We first show an exact reduction from $(Q,\pi,\Delta)$ to $(Q^+,\pi,\Delta^+)$.
    Carmeli and Kr\"oll~\cite[Claim 1 of Theorem 2]{carmeli2020enumeration} show an exact reduction from $(Q(X),\Delta)$ to $(Q^+(X^+),\Delta^+)$ without discussing the answer order.
    In this reduction, the bijection $\tau$ keeps the assignments to the free variables of $Q$ as is and projects away the assignments for the variables $X^+\setminus X$.
    Since $\pi$ only depends on $X$, the answer order is maintained.

We next show an exact reduction from $(Q^+,\pi,\Delta^+)$ to $(Q,\pi,\Delta)$.
    Carmeli and Kr\"oll~\cite[Claim 2 of Theorem 2]{carmeli2020enumeration} show an exact reduction from $(Q^+,\Delta^+)$ to $(Q,\Delta)$ without discussing the answer order.
    In this reduction, the bijection $\tau$ keeps the assignments to the free variables of $Q$ as is and only extends the mappings with assignments for the variables $X^+\setminus X$.
    Since $\pi$ only involves variables of $X$, the answer order is maintained.
\end{proof}

\end{toappendix}

Building on a reduction by Carmeli and Kr\"oll~\cite{carmeli2020enumeration}, we can reduce $(Q,\pi,\Delta)$ to $(Q^+,\pi,\Delta^+)$. The latter represents the same problem as $(Q^+,\pi^+,\Delta^+)$ by \Cref{lem:reordering-validity}. Finally, as any database that satisfies $\Delta^+$ also satisfies $\emptyset$, this can be reduced to $(Q^+,\pi^+,\emptyset)$. Overall, we get an exact reduction from a problem to its $\Delta$-reordered extension without FDs.

\begin{theoremrep}\label{thm:easy-extension}
    Given a self-join-free CQ $Q$, a variable order $\pi$, and a set of FDs $\Delta$, there is an exact reduction from $(Q,\pi,\Delta)$ to $(Q^+,\pi^+,\emptyset)$.
\end{theoremrep}
\begin{appendixproof}
    We reduce $(Q,\pi,\Delta)$ to $(Q^+,\pi,\Delta^+)$ using \Cref{lem:lex-reduct}. The latter represents the same problem as $(Q^+,\pi^+,\Delta^+)$ by \Cref{lem:reordering-validity}. 
    Finally, we apply \Cref{obs:remore-fds} for a reduction from $(Q^+,\pi^+,\Delta^+)$ to $(Q^+,\pi^+,\emptyset)$.
\end{appendixproof}

\Cref{thm:easy-extension} implies an algorithm for $(Q,\pi,\Delta)$ with the same time guarantees we have for $Q^+$ and $\pi^+$ without FDs.

\begin{corollary}\label{cor:ext-alg}
    Given a self-join-free join query $Q$, an ordering $\pi$ of its variables, and a set $\Delta$ of FDs, let~$\iota$ be the incompatibility number of their reordered extension $Q^+$ and $\pi^+$.
	Lexicographic direct access for $(Q,\pi, \Delta)$ is possible with $O(|D|^{\iota})$ preprocessing time and $O(\log(|D|))$ access time.
\end{corollary}
This algorithm is obtained by applying the reduction from \Cref{thm:easy-extension}, and then the algorithm from \Cref{thm:bringmann-et-al}.

Is this the most efficient way of solving $(Q,\pi,\Delta)$?
In case $\Delta$ contains only unary FDs, we answer this question positively using an exact reduction in the opposite direction in Section~\ref{sec:extension-unary-FDs}; we answer it negatively for general FDs in Section~\ref{sec:challenges}.

\subsection{Unary FDs}\label{sec:extension-unary-FDs}

\begin{toappendix}
    
\begin{lemma}\label{lem:unary-eliminate-fds}
    Given a self-join-free CQ $Q$, a variable order $\pi$, and a set of unary FDs $\Delta$, there is an exact reduction from $(Q^+,\pi^+,\emptyset)$ to $(Q^+,\pi^+,\Delta^+)$.
\end{lemma}
\begin{proof}
    We follow a similar construction to that of \cite[Lemma 8.6]{carmeli2023linearp} and show that it preserves the lexicographic order $\pi^+$.
    Given a database $D$ that is not required to conform to any FDs, we construct a database $D^+$ conforming to $\Delta^+$ as follows.
    
    Given a variable $v$, denote by $\vec{imp}_v$ the sequence of all variables that are transitively implied by $v$ in the order in which they appear in $\pi^+$ 
    (if $v$ does not appear in $\pi$, take an arbitrary order).
    By the definition of the extension, for every variable $v$, every atom of $Q^+$ that contains $v$ also contains all $\vec{imp}_v$ variables.
    For every fact in the relation that corresponds to such an atom, 
    we replace the $v$-value with the concatenation of $\vec{imp}_v$-values. 
    These concatenated domain values are assumed to be sorted lexicographically according to $\vec{imp}_v$.
    This construction can be done in linear time, and the resulting database satisfies the FDs.
    
    We now claim that this construction preserves the answers through a bijection.
    Note that for every free variable $v$, by the definition of the extension, the variables in $\vec{imp}_v$ are all free.
    Every answer to the original problem gives an answer over our construction by assigning every free variable $v$ with the concatenation of the assignments of $\vec{imp}_v$.
    Every answer over our construction gives an answer to the original problem by keeping only the value that corresponds to the original variable for every free variable.

    Since $\pi^+$ contains right after $v$ all variables of $\vec{imp}_v$ that do not appear before $v$, the answer order is preserved.
\end{proof}

\end{toappendix}

In case $\Delta$ contains only unary FDs, we can show an exact reduction from $(Q^+,\pi^+,\emptyset)$ to $(Q^+,\pi^+,\Delta^+)$. Let us demonstrate this reduction using an example.    
Consider $\Delta=\{x_1\rightarrow x_3\}$, the join query with the body $R(x_1,x_3),S(x_3,x_2)$ which is the same as its extension, and the order  $\pi=(x_1,x_2,x_3)$ which becomes $\pi^+=(x_1,x_3,x_2)$.
For every fact $R(c_1,c_3)$, the construction will introduce the fact $R((c_1,c_3),c_3)$, whereas the facts in $S$ stay the same.
An answer $(c_1,c_3,c_2)$ to the original instance corresponds to an answer $((c_1,c_3),c_3,c_2)$ to the constructed instance. The answers we obtain conform to the order $\pi^+$. We remark that we would not be able to have a similar construction with the order $\pi$ because we would need the answers to be sorted by $c_2$ before they are sorted by $c_3$, but $c_3$ appears as part of the concatenated value of $x_1$ while $c_2$ does not.

Since $(Q^+,\pi^+,\Delta^+)$ and $(Q^+,\pi,\Delta^+)$ represent the same problem, and since Carmeli and Kr\"oll~\cite{carmeli2020enumeration} provide a reduction from $(Q^+,\pi,\Delta^+)$ to $(Q,\pi,\Delta)$, we get the following.

\begin{theoremrep}\label{thm:unary-exact}
    Given a self-join-free CQ $Q$, a variable order $\pi$, and a set of unary FDs $\Delta$, there are exact reductions in both directions between $(Q,\pi,\Delta)$ and $(Q^+,\pi^+,\emptyset)$.
\end{theoremrep}
\begin{appendixproof}
    The reduction from $(Q,\pi,\Delta)$ to $(Q^+,\pi^+,\emptyset)$ is given by \Cref{thm:easy-extension}.
    For the other direction, \Cref{lem:unary-eliminate-fds} reduces $(Q^+,\pi^+,\emptyset)$ to $(Q^+,\pi^+,\Delta^+)$, \Cref{lem:reordering-validity} shows that $(Q^+,\pi^+,\Delta^+)$ and $(Q^+,\pi,\Delta^+)$ represent the same problem, and \Cref{lem:ext-to-orig} reduces $(Q^+,\pi,\Delta^+)$ to $(Q,\pi,\Delta)$.
\end{appendixproof}

\Cref{thm:unary-exact} implies that, in case of unary FDs, the complexity of lexicographic direct access (as well as any other problem closed under exact reductions, e.g., lexicographically-ordered enumeration) is the same as the complexity for the reordered extension without FDs, as long as the preprocessing is at least linear.
As the complexity of lexicographic direct access for join queries is well understood, see Theorem~\ref{thm:bringmann-et-al}, we can deduce a complexity result in the presence of unary FDs.

\begin{corollary}\label{cor:unary-inc}
    Given a self-join-free join query $Q$, an ordering $\pi$ of its variables, and a set of unary FDs $\Delta$, let~$\iota$ be the incompatibility number of their reordered extension $Q^+$ and $\pi^+$.
		\begin{itemize}
		  \item Lexicographic direct access for $(Q,\pi,\Delta)$ is possible with $O(|D|^\iota)$ preprocessing time and $O(\log(|D|))$ access time.
		  \item There is no constant $\varepsilon > 0$ such that for all $\delta > 0$ Lexicographic direct access for $(Q,\pi,\Delta)$ is possible with preprocessing time $O(|D|^{\iota-\varepsilon})$ and access time $O(|D|^\delta)$, assuming the Zero-Clique Conjecture.
		\end{itemize}
\end{corollary}

\subsection{Challenges with general FDs}\label{sec:challenges}

Let us notice that we cannot have a similar construction to that of Section~\ref{sec:extension-unary-FDs} for general FDs. Consider the FD $y_1,y_2\rightarrow z$ and the join query with the body $R(y_1,x_1),S(y_1,y_2,z),T(y_2,x_2)$, which is equal to its extension. A similar construction would copy the $z$ value to $y_1$ or $y_2$, but since $z$ does not appear in all atoms where $y_1$ and $y_2$ appear, such a construction cannot be done in linear time and would not constitute an exact reduction.
In fact, this is not just a problem with the specific construction: we will see in \Cref{ex:controller-query} that \Cref{cor:unary-inc} does not hold for general FDs.
That is, reordered extensions do not capture the direct access complexity in general.

A natural question is whether reordered extensions capture the complexity in the restricted case of acyclic queries (and general FDs).
We remark that, for enumeration with constant delay after linear preprocessing, $\Delta$-extensions do capture all tractable cases for acyclic CQs, as they do for unary FDs, even though the general case is open~\cite{carmeli2020enumeration}.
One could expect that a proof in the spirit of that used for enumeration would also apply here. However, we found that this proof does not generalize to our case, and in fact, we answer this question negatively, as shown by the following acyclic example.

\begin{example}\label{ex:controller-query}
Consider $Q(v_1, v_2, v_3, v_4, v_5) \leftarrow R_1(v_1,v_4),R_3(v_4,v_3,v_5),R_2(v_5,v_2)$, with the FD $v_4,v_5\rightarrow v_3$, and the variable order $\pi=(v_1,v_2,v_3,v_4,v_5)$.

The reordered extension is the same as the original problem.
The incompatibility number is $\iota(Q,\pi)=\rho^*(B_3) = 3$ because $B_3=\{v_1,v_2,v_3\}$. Thus, the optimal preprocessing time required for direct access with polylog access time for the reordered extension $(Q^+,\pi^+,\emptyset)$ is $\Theta(|D|^3)$ up to subpolynomial factors by Theorem~\ref{thm:bringmann-et-al}.
However, the FDs reduce the complexity in this case, and there is an algorithm for direct access for $(Q^+,\pi^+,\Delta^+)=(Q,\pi,\Delta)$ with preprocessing time $O(|D|^{2})$, as we explain next.

Consider the following algorithm.
First, build a lookup table according to the FD in $R_3$. That is, given a pair $(a_4,a_5)$, the lookup table searches for a value $a_3$ such that $(a_4,a_3,a_5)\in R_3$. Note that, due to the FD, at most one such value exists. 
Then, for every $(a_1,a_4) \in R_1$ and $(a_5,a_2)\in R_2$, search for $(a_4,a_5)$ in the lookup table, and if a value $a_3$ is found, store $(a_1,a_2,a_3,a_4,a_5)$ as an answer to $Q$.
Building the lookup table takes linear time, and it is used at most $|D|^2$ times, with a constant time for each lookup.
Overall, the algorithm takes $O(|D|^2)$ time to produce all the answers. We can now sort all answers, store them in an array, and support direct access with constant access time.

The preprocessing time $O(|D|^2)$ is tight for our example (up to sub-polynomial factors) since we can use our query to solve $\qstar_2$, as we sketch next.
For every fact $R_1^\star(a,c)$, introduce the facts $R_1(a,c)$ and $R_3(c,c,c)$. For every fact $R_2^\star(b,c)$, introduce the fact $R_2(c,b)$. This construction satisfies the FD as $v_3$, $v_4$, and $v_5$ are always equal in $R_3$. As this is a linear-time exact reduction, direct access to our example cannot be faster than for $\qstar_2$, which, by \Cref{thm:star-hard}, requires essentially quadratic preprocessing.
The next two sections devise general algorithms and lower bounds that correctly capture the complexity of this example.
\end{example}

\section{An Algorithm Based on the Polymatroid Bound}\label{sct:algorithm}

In this section, we will show our main algorithmic result, an algorithm for direct access that uses PANDA. To this end, we define a generic way of defining width measures based on disruption-free decompositions and specialize this framework to apply the polymatroid bound. We then give an algorithm whose runtime is determined by the resulting measure. Finally, we show that this algorithm should be applied on a $\Delta$-reordering to get better guarantees.

\subsection{Width Measures for Disruption-Free Decompositions}\label{sec:abstractwidth}

This section introduces a width measure based on the polymatroid bound that captures the complexity we obtain when using PANDA over a decomposition for the preprocessing for direct access. We start in a slightly more abstract way, introducing a generic width measure which we then reuse in Section~\ref{sec:lower} for lower bounds. 

The underlying idea is the same as in~\cite{bringmann2024stars}: we want to transform the input query into an acyclic query compatible with the given variable order (i.e., the incompatibility number is $1$) and construct a matching database, so we can then simply run a direct access algorithm from~\cite{carmeli2023linearp}. The new query has atoms whose variable sets are the bags of the disruption-free decomposition, and the complexity bottleneck lies in materializing relations for these atoms. 
Without FDs, bounding this complexity by the AGM-bound of the bags using a worst-case optimal join algorithm yields an optimal approach~\cite{bringmann2024stars}. For the case with FDs, there are generally no known tight bounds on join sizes and no worst-case optimal join algorithms, so we introduce a generic way of measuring the size of a bag which we then specialize for different applications. We now explain the technical details.

\begin{definition}\label{def:fwidth}
Let $\pi$ be a variable order of a join query $Q(X)$, and let $f:2^X\rightarrow \mathbb R_+$. We define the \emph{disruption-free $f$-width} of $Q$ and $\pi$ as $\width{f}(Q, \pi):=\max_{i\in n} f(B_i)$, where $\{X, \{B_1, \ldots, B_n\}\}$ is the disruption free decomposition of $Q$ according to $\pi$.
\end{definition}
We remark that this essentially the same as the definition of $f$-width proposed by Adler~\cite{Adler06}.
Notice that $\width{\fec{Q}}(Q,\pi)$ is the incompatibility number of $Q$ and $\pi$, as defined in~\Cref{sec:prelim}.
Thus, the algorithm from \Cref{cor:ext-alg} runs with preprocessing time $O(|D|^{\width{\fec{Q^+}}(Q^+,\pi^+)})$.

In this section, we will mostly be interested in the case where the function $f$ from \Cref{def:fwidth} is the exponent of the polymatroid bound, in which case we call the resulting width measure  $\width{\pb{Q}{\Delta}}(Q, \pi)$ the \emph{disruption-free polymatroid bound} of $Q$, $\Delta$ and $\pi$.

\begin{example}\label{ex:running-polymatroid}
Consider \Cref{ex:controller-query} again. We claim that~$\width{\pb{Q}{\Delta}}(Q, \pi)=2$. Set $X:= \{v_1, v_2, v_3, v_4, v_5\}$. First note that $B_3 = \{v_1, v_2, v_3\}$. It is readily checked that the set function defined by $h(S):= |S\cap \{v_1, v_2\}|$ is a polymatroid, guarded by $Q$, and respecting the FD $v_4, v_5 \rightarrow v_3$. Moreover, $h(B_3) = 2$, so $\width{\pb{Q}{\Delta}}(Q, \pi) \ge h(B_3) \ge 2$.

Next, consider any polymatroid $h$ guarded by $Q$ and respecting the FD. We will show that $h(B_i)\le 2$ for every bag $B_i$, which proves that $\width{\pb{Q}{\Delta}}(Q, \pi)\le 2$. First, by submodularity, we get $h(\{v_3, v_4, v_5)\}) + h(X\setminus \{v_3\}) \ge h(X)+ h(\{v_4, v_5\})$.
Since $h$ respects the FD, $h(\{v_4, v_5\}) = h(\{v_3, v_4, v_5\})$.
So overall, $h(X\setminus \{v_3\}) \ge h(X)$.
Let $B_i$ be any bag of the disruption-free decomposition, then
    $h(B_i) \le h(X) \le h(X\setminus \{v_3\}) \le h(\{v_1, v_4\}) + h(\{v_5, v_2\})\le 2$,
where the first inequality is by monotonicity, the third by submodularity, and the fourth because $h$ is guarded by $Q$.
Overall, we get that $\width{\pb{Q}{\Delta}}(Q, \pi)= 2$.
\end{example}

\subsection{The Algorithm}\label{sec:algorithm}

We get the following runtime bounds based on the disruption-free polymatroid bound.

\begin{theorem}\label{thm:generalupper}
Let $Q(X)$ be a join query, $\pi$ a variable order, and $\Delta$ a set of FDs. Then there is an algorithm that, given a database $D$ that respects $\Delta$, allows direct access on $Q(D)$ with preprocessing time $O(|D|^{\width{\pb{Q}{\Delta}}(Q, \pi)}\polylog(|D|))$ and access time $O(\polylog(|D|))$.\footnote{A very recent preprint~\cite{AboKhamisNS25} offers an improved version of PANDA with a tighter output bound. Using it would improve the $\polylog(|D|)$ factors in both the preprocessing and the access time to $\log(D)$.}
\end{theorem}
\begin{proof}
Denote by $Q'$ the join query represented by the disruption-free decomposition of $Q$ according to $\pi$ (containing an atom for every bag).
We will show how to build a database $D'$ such that $Q(D)=Q'(D')$ in $O(|D|^{{\width{\pb{Q}{\Delta}}(Q, \pi)}}\polylog(|D|))$ time. 
As discussed in the preliminaries, we can then use a direct access algorithm from \cite{carmeli2023linearp} for $Q'(D')$ with linear preprocessing and logarithmic access time.
This acts as a direct access algorithm for $Q(D)$ with the required time bounds.

Consider a bag $B_i$ of the decomposition, and consider a CQ $Q_i(B_i)$ representing this bag, with the same body as $Q$ and the bag variables as the free variables.
We can compute a relation $S_i \supseteq Q_i(D)$ using PANDA in time $|D|^{\pb{Q}{\Delta}(B_i)}\polylog(|D|)$. This fits within the required time because $\pb{Q}{\Delta}(B_i)\le {{\width{\pb{Q}{\Delta}}(Q, \pi)}}$.
We then filter $S_i$ by all atoms of $Q$ that are contained in the bag. That is, if $R_i(X_i)$ is an atom of $Q$ with $X_i\subseteq B_i$, remove from $S_i$ all facts that, projected to $X_i$, do not appear in $R_i^D$.
Denote by $D'$ the obtained database, with a relation $S_i$ for each bag $B_i$ of the decomposition. 

We have that $Q(D)\subseteq Q'(D')$ because PANDA gives a superset of the answers $Q(D)$ projected to the bag variables, and the filtering step does not remove answers of $Q$. We have that $Q(D)\supseteq Q'(D')$ because the relations comprising $D'$ are filtered by all atoms of $Q$ (since every atom of $Q$ is contained in some bag). 
\end{proof}

\begin{example}
\Cref{thm:generalupper} provides an algorithm for \Cref{ex:controller-query} with quadratic preprocessing since $\width{\pb{Q}{\Delta}}(Q, \pi)=2$, as seen in \Cref{ex:running-polymatroid}.
\end{example}

\subsection{The Effect of Reorderings}

In~\Cref{sec:algorithm}, we have seen a direct access algorithm that is exponential in~$\width{\pb{Q}{\Delta}}(Q, \pi)$.  From \Cref{lem:reordering-validity}, we know that we can reorder the variables specifying a desired lexicographic order in certain ways while maintaining the same answer order. The obvious question is then whether we should always do so, or whether leaving the variable order unchanged can yield faster preprocessing. This section shows that $\Delta$-reorderings may decrease the disruption-free polymatroid bound, they never increase it, and every $\Delta$-reordering yields the same bound. Thus, it is always a good idea to replace a variable order with an arbitrary $\Delta$-reordering before applying the algorithm of \Cref{thm:generalupper}. 

\begin{proposition}\label{prop:improvementpolymatroid}
Let $Q$ be a join query, $\pi$ a variable order, and $\Delta$ a set of FDs.
If $\pi'$ and $\pi''$ are $\Delta$-reorderings of $\pi$, then $\width{\pb{Q}{\Delta}}(Q, \pi'')=\width{\pb{Q}{\Delta}}(Q, \pi') \le \width{\pb{Q}{\Delta}}(Q, \pi)$.
\end{proposition} 

Let us first show that reordering can decrease $\width{\pb{Q}{\Delta}}(Q, \pi)$, improving the runtime of \Cref{thm:generalupper}. The following example shows that this runtime gain is unbounded.

\begin{example}
    Consider $Q^\star_k$ with $\pi= (x_1, \ldots, x_k, x_{k+1})$ and $\Delta=\{x_1 \rightarrow x_{k+1}\}$. We claim that $\width{\pb{Q_k^\star}{\Delta}}(Q_k^\star, \pi)\ge k-1$. To see this, note first that the set function defined by $h(S):= |S\cap \{v_2, \ldots, v_k\}|$ is a polymatroid, respecting the FD, and guarded by $Q^\star_k$. Moreover, we have that $B_{k+1} = \{v_1, \ldots, v_{k+1}\}$, and $h(B_{k+1}) = k-1$. So, $\width{\pb{Q_k^\star}{\Delta}}(Q_k^\star, \pi)\ge \pb{Q_k^\star}{\Delta}(B_{k+1})\ge k-1$.
    After reordering to $\pi'=(x_1, x_{k+1}, x_2, \ldots, x_k)$, we get that, other than $B_1=\{x_1\}$, all bags are of the form $\{x_i, x_{k+1}\}$. Since $Q^\star_k$ contains the atoms $R_i^\star(x_i, x_{k+1})$, any guarded polymatroid has $h(\{x_i,x_{k+1}\})\le 1$, and it follows that $\width{\pb{Q_k^\star}{\Delta}}(Q_k^\star, \pi')\le 1$. So, the preprocessing time from \Cref{thm:generalupper} decreases from roughly $|D|^{k-1}$ to quasilinear due to $\Delta$-reordering.
\end{example}

Instead of proving Proposition~\ref{prop:improvementpolymatroid} directly, we show a more abstract version that is also useful for \Cref{sec:lower}. To this end, we define another property of set functions.
Let $Q(X)$ be a join query and $\Delta$ a set of FDs. We call a set function $f:2^X\rightarrow \mathbb R_+$ \emph{$\Delta$-stable} if for every FD $Y\rightarrow z$ in $\Delta$ and every set $Y\subseteq S\subseteq X$, we have that $f(S) = f(S\cup \{z\})$.
Since $\pb{Q}{\Delta}$ is $\Delta$-stable, all properties that we show for general $f$-stable functions hold for $\pb{Q}{\Delta}$.
\begin{lemmarep}\label{lem:polymatroidstable}
Let $Q(X)$ be a join query and let $\Delta$ be a set of FDs.  Then, $\pb{Q}{\Delta}$ is $\Delta$-stable.
\end{lemmarep}
\begin{appendixproof}
Consider an FD $Y\rightarrow z$ in $\Delta$ and a set $Y\subseteq S\subseteq X$.
It is known that for any polymatroid $h:2^V \rightarrow \mathbb R_+$, all sets $A\subseteq B\subseteq V$ and $v\in V$, we have $h(A\cup \{v\})-h(A) \ge h(B\cup \{v\}) - h(B)$~\cite[Chapter~44]{Schrijver03}. Setting $A=Y, B=S$ and $v=z$, it follows that $h(Y\cup \{z\}) - h(Y) \ge h(S\cup \{z\}) - h(S)$. By monotonicity, $h(S\cup \{z\}) - h(S)\ge 0$.
If $h$ also respects the FD $Y\rightarrow z$, we have that $h(Y\cup \{z\}) - h(Y) = 0$, and thus $h(S\cup \{z\}) = h(S)$.
It follows that $\pb{Q}{\Delta}(S\cup \{z\}) = \pb{Q}{\Delta}(S)$.  
\end{appendixproof}

We can show that, for a $\Delta$-stable set function $f$, the $f$-width of a $\Delta$-reordering is at most that of the original order. We get that two $\Delta$-reorderings of the same order have the same $f$-width because they are $\Delta$-reorderings of each other.

\begin{proposition}\label{prop:general-reordering-equal-and-better}
Let $Q$ be a join query with FDs $\Delta$ and a variable order $\pi$. Let $\pi'$ and $\pi''$ be $\Delta$-reorderings of $\pi$. If $f:2^{\var(Q)}\rightarrow \mathbb R_+$ is monotone and $\Delta$-stable, then $\width{f}(Q, \pi'') = \width{f}(Q, \pi')\le \width{f}(Q, \pi)$.
\end{proposition}

Thus, the algorithm we propose in this section is obtained by taking a $\Delta$-reordering and then applying \Cref{thm:generalupper}. Given a join query $Q$, a set $\Delta$ of FDs, and a $\Delta$-reordering $\pi'$ of a desired variable order, the preprocessing obtained is $O(|D|^{\width{\pb{Q}{\Delta}}(Q, \pi')}\polylog(|D|))$.

\begin{toappendix}
    
\subsection{Proof of \Cref{prop:general-reordering-equal-and-better}}

We first show that for monotone and $\Delta$-stable functions $f$, reordering can only decrease the disruption-free $f$-width.
We start by proving a special case.
\begin{lemma}\label{lem:swapimprovement}
Let $Q$ be a join query, $\Delta$ a set of FDs, and $\pi = (v_1 \ldots v_n)$ a variable order. Let $\pi'= (v_1 \ldots v_{i-1} v_{i+1}v_iv_{i+2}\ldots v_n)$ such that there is $Y\subseteq \{v_1\ldots v_{i-1}\}$ where $Y\rightarrow v_{i+1}$ is an FD in $\Delta$.  If $f:2^{\var(Q)}\rightarrow \mathbb R_+$ is monotone and $\Delta$-stable, then $\width{f}(Q, \pi') \le \width{f}(Q, \pi)$.
\end{lemma}
\begin{proof}
For every $j$, let $B_j$ and $B_j'$ the bags as in definition of the disruption-free decompositions for $\pi$ and $\pi'$, respectively. Moreover, let $S_j$ and $S_j'$ be defined as the sets $S_j$ for $\pi$ and $\pi'$. We first observe that for $j \notin \{i, i+1\}$, we have $S_j= S_j'$ and thus $B_j= B_j'$. So,
the only bags that could lead to a difference between the disruption-free $f$-widths are those created by $v_i$ and $v_{i+1}$.

Consider first the case that $v_i$ and $v_{i+1}$ are not connected by a path in $Q[v_i, v_{i+1}, \ldots, v_n]$. Then, $S_i= S_i'$ and $S_{i+1}= S_{i+1}'$, and thus $B_i= B_i'$ and $B_{i+1}= B_{i+1}'$. It follows directly that the disruption-free $f$-width of $Q$ and $\pi$ is the same as that of $Q$ and $\pi'$.

If $v_i$ and $v_{i+1}$ are connected by a path in $Q[v_i, v_{i+1}, \ldots, v_n]$, then
$v_i$ and $v_{i+1}$ appear in the same connected component of $Q[v_i, v_{i+1}, \ldots, v_n]$, and this connected component is
$S_{i} = S_{i+1}'$. 
We claim that both $f(B_i')$ and $f(B_{i+1}')$ are at most $f(B_i)$, from which the lemma follows directly. 

We first show that $f(B_{i+1}') \le f(B_i)$. To this end, observe that, due to $S_{i} = S_{i+1}'$, we have $B_i\setminus \{v_i\} = B_{i+1}'\setminus \{v_{i+1}\}$. Moreover, since the FD $Y\rightarrow v_{i+1}$ is covered by a relation, we have that all variables in $Y$ are neighbors of $v_{i+1}$ in $Q$, so $Y\subseteq B_{i+1}'$ and thus by $\Delta$-stability, we have $f(B_{i+1}') = f(B_{i+1}'\setminus \{v_i+1\})= f(B_i\setminus \{v_i\}) \le f(B_i)$, where the inequality follows from monotonicity of $f$.

We next show that $f(B_{i}') \le f(B_i)$. First note that $v_{i+1}\in S_i$ since there is a path from $v_i$ to $v_{i+1}$ in $Q[v_i, v_{i+1}, \ldots, v_n]$. Since all variables in $Y$ are neighbors of $v_{i+1}$, we have $Y\subseteq B_i$. It follows that $f(B_i) = f(B_i \cup \{v_{i+1}\})$ since $f$ is $\Delta$-stable. We have $S_i'\subseteq S_i$, so $B_i'\subseteq B_i\cup \{v_{i+1}\}$, and thus $f(B_i')\le f(B_i\cup \{v_{i+1}\}) = f(B_i)$.
\end{proof}

We can now prove the general case using \Cref{lem:swapimprovement}.

\begin{lemma}\label{lem:improvementabstract}
Let $Q$ be a join query, $\pi$ a variable order, and $\Delta$ a set of FDs. Let $\pi'$ be a $\Delta$-reordering of $\pi$. If $f:2^{\var(Q)}\rightarrow \mathbb R_+$ is monotone and $\Delta$-stable, then $\width{f}(Q, \pi') \le \width{f}(Q, \pi)$
\end{lemma}
\begin{proof}
The idea is that we can simulate the construction of any $\Delta$-reordering by iteratively bubbling variables to the front, increasing the disruption-free $f$-width in none of the steps due to \Cref{lem:swapimprovement}.
More precisely, we claim that there exists a sequence $\pi=\pi_1\ldots,\pi_\ell=\pi'$ of variable orders such that $\pi_j$ and $\pi_{j+1}$ satisfy the conditions of \Cref{lem:swapimprovement} for all $j$. We conclude that $\width{f}(Q, \pi_{j+1}) \le \width{f}(Q, \pi_j)$, and overall $\width{f}(Q, \pi') \le \width{f}(Q, \pi)$.

Consider a variable order $\pi_j\neq \pi'$, let $t$ be the first position in which they differ, and let $z$ be the variable in position $t$ in $\pi'$. By the definition of the reordering, there is an FD $Y\rightarrow z$ with $Y$ appearing in the first $t-1$ positions of $\pi_j$. Define $\pi_{j+1}$ to be similar to $\pi_j$ but with the variable $z$ advanced by one position. Clearly, the conditions of \Cref{lem:swapimprovement} apply for $\pi_j$ and $\pi_{j+1}$.
Compared to $\pi_j$, $\pi_{j+1}$ either has a later first different position with respect to $\pi'$, or the same first different position $t$ and the variable in position $t$ in $\pi'$ is closer to $t$ in $\pi_{j+1}$ than in $\pi_j$. Thus, this sequence will eventually produce $\pi'$.
\end{proof}

We get that different reorderings have the same disruption-free $f$-width.

\begin{lemma}\label{lem:reorderingsallthesameabstract}
Let $Q$ be a join query, $\pi$ a variable order, and $\Delta$ a set of FDs. Let $\pi'$ and $\pi''$ be $\Delta$-reorderings of $\pi$. If $f:2^{\var(Q)}\rightarrow \mathbb R_+$ is monotone and $\Delta$-stable, then $\width{f}(Q, \pi') = \width{f}(Q, \pi'')$.
\end{lemma}
\begin{proof}
    We notice that $\pi'$ is a $\Delta$-reordering of $\pi''$ and vice-versa.
    Indeed, in the construction of $\Delta$-reorderings in Definition~\ref{def:FD-reordering}, there are two reasons to add a variable $z$ at a specific position to a reordering $\pi'$: either there is a set $Y$ already added to $\pi'$ such that $Y\rightarrow z$ or there is no such FD but $z$ is the next variable in $\pi$ that has not already been added to $\pi'$. We call the former type of variables \emph{dependent} and the latter type \emph{independent}. For every independent variable $v_j$, let us define the set $\dep(v_j)$ to contain all dependent variables that we have to add until we add the next independent variable. Clearly, $\dep(v_j)$ does not depend on the choices in the construction and is thus well-defined. Also, the orders $\pi'$ and $\pi''$ can only differ in the order inside the sets $\dep(y_j)$ while the positions of the independent variables are fixed, and any order inside the sets $\dep(y_j)$ is allowed in a reordering.
    Thus, $\pi'$ is a $\Delta$-reordering of $\pi''$ and vice-versa.
    
    By applying \Cref{lem:improvementabstract} twice, we get that $\width{f}(Q, \pi') \le \width{f}(Q, \pi'') \le \width{f}(Q, \pi')$, so overall $\width{f}(Q, \pi')=\width{f}(Q, \pi'')$.
\end{proof}

\Cref{prop:general-reordering-equal-and-better} is obtained by combining \Cref{lem:reorderingsallthesameabstract} and \Cref{lem:improvementabstract}.
\end{toappendix}

\section{Lower Bounds Based on FD-Aware Incompatibility}\label{sec:lower}

 In Section~\ref{sec:challenges}, we have seen that the extension approach cannot be used to show lower bounds for direct access when the FDs we consider are not unary. Thus, to show lower bounds in this setting, we have to develop new techniques. In this section, we do so by combining the decomposition approach used to show lower bounds for direct access without FDs~\cite{bringmann2024stars} with the coloring approach used to show size lower bounds for CQs under FDs~\cite{GottlobLVV12}.
 
\subsection{FD-Aware Incompatibility}

We define a generalization of the incompatibility number~\cite{bringmann2024stars} for queries with FDs.
To this end, we use a generalization of the fractional edge cover number from~\cite{GottlobLVV12}, \emph{the color number}.

Given a join query $Q$, a $k$-\emph{coloring} $\col: \var(Q) \rightarrow 2^{[k]}$ with $k\ge 1$ assigns a set of colors to each query variable. 
Given a set of variables $V$, we set $\col(V)=\bigcup_{v\in V}{\col(v)}$, and we say that $V$ \emph{contains} a color $c$ if $c\in\col(V)$.
Given a set of FDs, we say that a coloring $\col$ is \emph{valid} if for every FD $Y\rightarrow z$, we have that $\col(z)\subseteq \col(Y)$. 

\begin{definition}[Color Number~\cite{GottlobLVV12}]\label{def:color-num}
Given a self-join-free join query $Q$ with a distinguished set of variables $S$ and a set $\Delta$ of FDs, the \emph{color number} $\cn{Q}{\Delta}(S)$ is the maximum over all valid colorings of $Q$ of the number of colors appearing in $S$ divided by the maximum number of colors appearing in an atom.
That is, $\cn{Q}{\Delta}(S):=\max_{\col}{\frac{|\col(S)|}{\max_{R_i(X_i)\in \atoms(Q)}|\col(X_i)|}}$.
\footnote{The color number was originally defined for queries with self-joins~\cite{GottlobLVV12}. We consider here a simplified definition without self-joins.}
\end{definition}

We call the width measure~$\width{\cn{Q}{\Delta}}$, as specified by \Cref{def:fwidth}, the~\emph{FD-aware incompatibility number}.
We remark that, in case there are no FDs, the FD-aware incompatibility number coincides with the incompatibility number from~\cite{bringmann2024stars}:
without FDs, the fractional edge cover number and the color number coincide~\cite[Section~3.1]{GottlobLVV12}.

\begin{example}\label{ex:running-color}
    Consider the query from \Cref{ex:controller-query} again. Its FD-aware incompatibility number is $2$, given by $B_3$ and the coloring with $\col(v_1)=\{1\}$, $\col(v_2)=\{2\}$, and $\col(v_3)=\col(v_4)=\col(v_5)=\emptyset$.
    Without the FD, the incompatibility number is $3$, as can be seen by the coloring $\col'(v_1)=\{1\}$, $\col'(v_2)=\{2\}$, $\col'(v_3)=\{3\}$, and $\col'(v_4)=\col'(v_5)=\emptyset$. However, given the FD $v_4,v_5\rightarrow v_3$, $\col'$ is not valid.
\end{example}

\begin{lemmarep}\label{lem:incompatibilitystable}
Let $Q(X)$ be a self-join-free join query and $\Delta$ a set of FDs. Then, the color number $\cn{Q}{\Delta}$ is monotone and $\Delta$-stable.
\end{lemmarep}
\begin{appendixproof} 
    First, notice that the color number is monotone. Consider the sets $S_1\subseteq S_2\subseteq X$, and  let $\col'$ be an optimal coloring for $Q$, $S_1$, and $\Delta$. Then, \[C(Q,S_1, \Delta)=\frac{|\col'(S_1)|}{\max_{R_i(X_i)\in \atoms(Q)}|\col'(X_i)|}\le \frac{|\col'(S_2)|}{\max_{R_i(X_i)\in \atoms(Q)}|\col'(X_i)|}\le C(Q,S_2, \Delta).\] 

    Next, we show $\Delta$-stability.
    Consider an FD $Y\rightarrow z$ of $\Delta$ and a set $Y\subseteq S\subseteq X$.
    Let $\col$ be an optimal coloring for $Q$, $S\cup \{z\}$, and $\Delta$. Since $\col$ is valid, we have that $\col(z)\subseteq \col(Y)$. Thus, $|\col(S\cup \{z\})|=|\col(S)|$, and
    \[C(Q,S\cup \{z\}, \Delta)=\frac{|\col(S\cup \{z\})|}{\max_{R_i(X_i)\in \atoms(Q)}|\col(X_i)|}=\frac{|\col(S)|}{\max_{R_i(X_i)\in \atoms(Q)}|\col(X_i)|}\le C(Q,S, \Delta).\]
    The inequality in the other direction holds because the color number is monotone, so, overall, we get that $C(Q,S, \Delta)=C(Q,S\cup \{z\}, \Delta)$.
\end{appendixproof}

As a direct consequence of \Cref{prop:general-reordering-equal-and-better} and \Cref{lem:incompatibilitystable}, it is well-defined to refer to the FD-aware incompatibility number of a $\Delta$-reordering.
\begin{proposition}\label{prop:reorderingsallthesame}
Let $Q$ be a self-join-free join query with FDs $\Delta$ and a variable order $\pi$. Let $\pi'$ and~$\pi''$ be $\Delta$-reorderings of $\pi$. Then, $\width{\cn{Q}{\Delta}}(Q, \pi') =\width{\cn{Q}{\Delta}}(Q, \pi'')$.
\end{proposition}

\subsection{Hardness Proof}

We will show that the FD-aware incompatibility number of a $\Delta$-reordering gives a lower bound for direct access.
As in~\cite{bringmann2024stars}, we will use a reduction from star queries.
The following lemma shows how to use the hardness of $\qstar$ to deduce the hardness of any query that admits a coloring with certain properties.
We say that a coloring is \emph{connected} if, for every color, the set of variables assigned this color is connected.
\begin{toappendix}
\begin{lemma}\label{lem:nice-coloring}
    Let $Q$ be a CQ, $S$ a subset of its variables, and $\Delta$ a set of FDs. There is~a valid connected coloring witnessing its color number $\cn{Q}{\Delta}(S)$ such that $S$ contains all colors.
\end{lemma}
\begin{proof}
    Consider a valid coloring witnessing the color number.
    First, remove colors that do not appear in $S$. 
    Then, if there are several connected components of variables containing a color, remove this color from all but one component that includes a variable in $S$.
    Since every FD is contained in an atom, removing entire connected components does not affect the validity of the coloring.
    These two steps do not change the number of colors that appear in $S$, and they do not increase the number of colors assigned to an atom. Thus, they do not decrease the color number. Since we started with a coloring witnessing the maximum possible color number, we conclude that the color number of the modified coloring is equal to that of the original coloring. 
\end{proof}
\end{toappendix}
Given a coloring $\col$ and an ordering of variables $\pi$, we say that a variable $v$ \emph{introduces} a color~$c$ if $v$ is the first variable in $\pi$ such that $c \in \col(v)$.

\begin{lemmarep}\label{lem:star-reduction}
Let $Q$ be a self-join-free join query over a schema with FDs, $\pi$ an ordering of its variables, $1 <\iota\in\mathbb{R}$, and $2\le k\in\mathbb{N}$, such that there is a valid connected $(k+1)$-coloring of $Q$ satisfying:
\begin{itemize}
    \item The color $k+1$ shares an atom with every other color.
    \item The maximum number of colors out of $[k]$ that appear in an atom is $\frac{k}{\iota}$.
    \item The variable that introduces the color $k+1$ does not come before a variable that introduces another color.
\end{itemize}
If there exists $\eps > 0$ such that for all $\delta > 0$ there is a direct access algorithm for $Q$ and $\pi$ with preprocessing time $O(|D|^{\iota - \varepsilon})$ and access time $O(|D|^{\delta})$, 
then there exists $\eps'>0$ such that for all $\delta'>0$ there is a direct access algorithm for $\qstar_k$ with respect to a bad ordering with preprocessing time $O(|D^\star|^{k-\varepsilon'})$ and access time $O(|D^\star|^{\delta'})$.
\end{lemmarep}
\begin{proofsketch}
    The construction is similar to that used by Bringmann et al.~\cite[Lemma 17]{bringmann2024stars}; we repeat the ideas here for completeness.
    Let $\col: \var(Q) \rightarrow 2^{[k+1]}$ be the said coloring of $Q$.
    Given a database $D^\star$ for $\qstar_k$, we construct a database $D$ for $Q$.
    Given an atom $R(v_1,\ldots,v_n)$, let $c_1,\ldots,c_\ell$ be those colors of $[k]$ that appear in $\col(v_1)\cup\ldots\cup\col(v_n)$. We note that $\ell\le \frac{k}{\iota}$.
    Compute the join of $R_{c_1}(x_{c_1},x_{k+1}),\ldots,R_{c_\ell}(x_{c_\ell},x_{k+1})$.
    Every answer assignment $h$ to this join defines a function $f_h$ from colors to the domain of $D^\star$; we extend this to a function $f^+_h$ from a set of colors by setting $f^+_h(\{c_1,\ldots,c_m\})=(f_h(c_1),\ldots,f_h(c_m))$ assuming $c_1<...<c_m$, and $f^+_h(\emptyset)=\bot$; and we add to $R^D$ the fact obtained by this function: $R(f^+_h(\col(v_1)),\ldots,f^+_h(\col(v_n)))$. We assume that the domain is sorted lexicographically. The construction satisfies the FDs because the coloring is valid.
    It can be shown that the construction encodes the answers in $\qstar_k(D^\star)$ in a bad lexicographic order, so direct access to $Q(D)$ simulates direct access to $\qstar_k(D^\star)$.
\end{proofsketch}
\begin{proof}
The construction is similar to that used by Bringmann et al.~\cite[Lemma 17]{bringmann2024stars}; we repeat the ideas here for completeness.
    Let $\col: \var(Q) \rightarrow 2^{[k+1]}$ be the said coloring of $Q$.

    \textbf{Construction.}
    Given a database $D^\star$ for $\qstar$, we construct a database $D$ for $Q$.
    Given an atom $R(v_1,\ldots,v_n)$, let $c_1,\ldots,c_\ell$ be those colors of $[k]$ that appear in $\col(v_1)\cup\ldots\cup\col(v_n)$. We note that $\ell\le \frac{k}{\iota}$.
    Compute the join of $R_{c_1}(x_{c_1},x_{k+1}),\ldots,R_{c_\ell}(x_{c_\ell},x_{k+1})$.
    Every answer assignment $h$ to this join defines a function $f_h$ from colors to the domain of $D^\star$; we extend this to a function $f^+_h$ from a set of colors by setting $f^+_h(\{c_1,\ldots,c_m\})=(f_h(c_1),\ldots,f_h(c_m))$ assuming $c_1<...<c_m$, and $f^+_h(\emptyset)=\bot$; and we add to $R^D$ the fact obtained by this function: $R(f^+_h(\col(v_1)),\ldots,f^+_h(\col(v_n)))$. We assume that the domain is sorted lexicographically. The construction satisfies the FDs because the coloring is valid.

    \textbf{Correctness.}
    We claim that we get a bijection $\tau$ from $Q(D)$ to $\qstar(D^\star)$: if a coordinate corresponding to a color $c$ is assigned a value $d$ in an answer $a$ to $Q(D)$, then $\tau(a)$ assigns $d$ to $x_c$.
    First, since the coloring is connected, all occurrences of the same color get assigned the same value, and so $\tau$ is well-defined.
    Given an answer to $\qstar(D^\star)$, it is clear that it is obtained by our construction because the construction combines subjoins of $\qstar$.
    Given an answer to $Q(D)$, we argue that it yields an answer to $\qstar$. Indeed, given any atom $R_i(x_i,x_{k+1})$ of $\qstar$, the color $i$ shares an atom with the color $k+1$ by the conditions of this lemma, and hence there exists a relation in $Q$ that uses the atom $R_i(x_i,x_{k+1})$ in the subjoin defining it.
    It is left to notice that the construction produces the results in a bad ordering because the variable that introduces $k+1$ does not come before variables that introduce other colors. If the same variable introduces both $k+1$ and another color $c$, our construction ensures that the values are sorted first by the assignment to $x_c$ and then by the assignment to $x_{k+1}$.

    \textbf{Time complexity.}
    Since every atom contains at most $\frac{k}{\iota}$ colors out of $[k]$, every relation of $D$ is obtained by the join of at most $\frac{k}{\iota}$ relations of $D^\star$. Thus,  $|D|=O(|D^\star|^{\frac{k}{\iota}})$.
    Given an algorithm for $Q$ with preprocessing time $O(|D|^{\iota - \eps})$, we set $\eps'=\min(k-\frac{k}{\iota},\frac{k\eps}{\iota})$. With this choice, we get that the construction time for the database $D$ is $O(|D^\star|^{\frac{k}{\iota}})\le O(|D^\star|^{k-\eps'})$, and the preprocessing time of this algorithm over our construction is 
    $O(|D|^{\iota - \eps})\le O(|D^\star|^{\frac{k}{\iota}(\iota-\eps)})\le O(|D^\star|^{k-\eps'})$.
    Then, given $\delta'>0$, we set $\delta=\frac{\delta'\iota}{k}$ and get that  the access time is $O(|D|^\delta)\le O(|D^\star|^{\frac{k}{\iota}\frac{\delta'\iota}{k}})\le O(|D^\star|^{\delta'})$.
\end{proof}

It remains to find the coloring used by \Cref{lem:star-reduction}.

\begin{lemmarep}\label{lem:coloring-construction}
    Let $Q$ be a self-join-free join query over a schema with FDs $\Delta$, and $\pi$ an ordering of its variables. Let $\pi'$ be a $\Delta$-reordering of $\pi$.
    If $\width{\cn{Q}{\Delta}}(Q, \pi')>1$, then there exists $2\le k\in\mathbb{N}$ such that there is a valid connected $(k+1)$-coloring of $Q$ satisfying:
\begin{enumerate}
    \item The color $k+1$ shares an atom with every other color. \label{property:touching-center}
    \item The maximum number of colors out of $[k]$ that appear in an atom is $\frac{k}{\width{\cn{Q}{\Delta}}(Q, \pi')}$.\label{property:bounded-colors-in-atom}
    \item The variable that introduces the color $k+1$ does not come before a variable that introduces another color.\label{property:center-last}
\end{enumerate}
\end{lemmarep}
\begin{proofsketch}
One can show that there is a valid connected coloring $\col: \var(Q) \rightarrow 2^{[k]}$ for some integer $k$ such that $|\col(B)| = k$ for some bag $B$ of the disruption-free decomposition and $\width{\cn{Q}{\Delta}}(Q, \pi') = \frac{k}{\max_{R(X)\in \atoms(Q)}|\col(X)|}$. Let $v_t$ be the last variable in $\pi$ that introduces a color, and let this color w.l.o.g.~be $k$.
We define $B_t$ and $S_t$ as in the definition of a disruption-free decomposition: $B_t$ is the bag obtained when treating $v_t$ (i.e., the bag in which $v_t$ is the last variable with respect to $\pi$), and $S_t$ is the connected component of $Q[v_t,\ldots,v_n]$ that contains $v_t$. Add a fresh color $k+1$ to all color sets $\col(v)$ with $v\in S_t$, and call the resulting coloring~$\col'$. One can show that $\col'$ has all claimed properties.
\end{proofsketch}
\begin{proof}
Use \Cref{lem:nice-coloring} to get a valid connected coloring $\col: \var(Q) \rightarrow 2^{[k]}$ for some integer $k$ such that $|\col(B)| = k$ for some bag $B$ of the disruption-free decomposition and $\width{\cn{Q}{\Delta}}(Q, \pi') = \frac{k}{\max_{R(X)\in \atoms(Q)}|\col(X)|}$. Let $v_t$ be the last variable in $\pi$ that introduces a color, and let this color w.l.o.g.~be $k$.
We define $B_t$ and $S_t$ as in the definition of a disruption-free decomposition: $B_t$ is the bag obtained when treating $v_t$ (i.e., the bag in which $v_t$ is the last variable with respect to $\pi$), and $S_t$ is the connected component of $Q[v_t,\ldots,v_n]$ that contains $v_t$. Add a fresh color $k+1$ to all color sets $\col(v)$ with $v\in S_t$, and call the resulting coloring~$\col'$.

\begin{claim}
    $\col'$ is a valid coloring.
\end{claim}
\begin{claimproof}
Since $\col'$ is constructed from the valid coloring $\col$ by adding the fresh color $k+1$ to every variable in $S_t$, we only need to show that for every FD $X\rightarrow y$ such that $y\in S_t$, there is a variable $x\in X$ in $S_t$.

Denote by $x$ the last variable of $X$ in $\pi$.
Let $u_x$ be the last variable in $\pi$ implied by the variables in the prefix of $\pi$ ending with $x$. In particular, $x\le_\pi u_x$. Due to the FD, we have that $y\le_\pi u_x$. Since $y\in S_t$, we have that $y\ge_\pi v_t$. So overall, $v_t\le_\pi u_x$.
Since $v_t$ introduces the color $k$ in the valid coloring $\col$, it cannot be implied by a set of variables coming before it (as the color $k$ must appear in one of the implying variables).
As $\pi$ is consistent with the FDs, this means that we cannot have that $x<_\pi v_t\le_\pi u_x$.
We conclude that $v_t\le_\pi x$.

Since there must be an atom containing all variables of the FD, $x$ and $y$ are neighbors, and since $y\in S_t$, we conclude that $x\in S_t$ as well.
\end{claimproof}

It is clear from the construction that the coloring $\col'$ is connected.
In $\col'$, as in $\col$, the maximum number of colors from $[k]$ in an atom is $\frac{k}{\width{\cn{Q}{\Delta}}(Q, \pi')}$, as required by \Cref{property:bounded-colors-in-atom} of the lemma. 
By construction, the color $k+1$ is introduced in $v_t$ which also introduces the last color of $[k]$. This proves \Cref{property:center-last} of the lemma.
We will show next that every color of $[k]$ appears in a variable of the bag $B_t$. We also have that every variable of $B_t$ is a neighbor of a variable in $S_t$ by definition of the bag. Since we assigned the color $k+1$ to every variable in $S_t$, every color shares an atom with $k+1$, as required by \Cref{property:touching-center} of the lemma.
It is left to prove the following claim.

\begin{claim}\label{clm:allcolors}
 $[k]\subseteq \col(B_t)$.
\end{claim}
\begin{claimproof}
    We define a sequence of bags $B_{i_1},\ldots,B_{i_\ell}$ such that $B_{i_1}=B$ is the bag from the definition of $\col$, $B_{i_{j+1}}$ is the bag created by the second largest element of $B_{i_j}$, and $B_{i_\ell}=B_t$ is the bag created by $v_t$. We will prove that this sequence is well-defined and that $[k]\subseteq \col(B_{i_j})$ for every $j\in[\ell]$.

    First, note that $\col(B_{i_1})=[k]$ by choice of $\col$. Next, assume that $[k]\subseteq \col(B_{i_j})$ and $B_{i_j}\neq B_t$.
    Let $v_{i_j}$ be the variable creating $B_{i_j}$ (and thus the largest variable in $B_{i_j}$).

    First, we claim that $v_{i_j} >_\pi v_t$. To see this, 
    observe that $B_{i_j}$ contains the color $k$, which is introduced in $v_t$. It follows that $B_{i_j}$ must contain a variable $\ge_\pi v_t$. Since $B_{i_j}$ only contains variables $\le_\pi v_{i_j}$ by definition, we deduce that $v_{i_j}\ge_\pi v_t$. Since we are in the case that $B_{i_j}\ne B_t$, we have that $v_{i_j} \neq v_t$ and thus $v_{i_j} >_\pi v_t$, as claimed.
    
    Let $c\in\col(v_{i_j})$.
    Since $v_{i_j}>_\pi v_t$ and $v_t$ introduces the last color, there exists a variable $u<_\pi v_{i_j}$ with $c\in\col(u)$. Since the coloring $\col$ is connected, there is a path from $v_{i_j}$ to $u$ on which every variable contains $c$. Let $u'$ be the first variable on this path such that $u'<_\pi v_{i_j}$. By definition of the decomposition, $B_{i_j}$ contains $u'$.
    Since this is true for every color in $\col(v_{i_j})$, we get that $\col(v_{i_j})\subseteq \col(B_{i_j}\setminus\{v_{i_j}\})$. It follows that $\col(B_{i_j}\setminus\{v_{i_j}\})=\col(B_{i_j})\supseteq [k]$.
    Hence, $B_{i_j}\setminus\{v_{i_j}\}\neq\emptyset$, so $B_{i_j}$ has a second largest element, and $B_{i_{j+1}}$ is well-defined. 
    By definition of the decomposition, $B_{i_{j+1}}$ must contain $B_{i_j}\setminus\{v_{i_j}\}$, so $[k]\subseteq\col(B_{i_j}\setminus\{v_{i_j}\})\subseteq\col(B_{i_{j+1}})$.

    Since the query is finite, the sequence $B_{i_1},\ldots,B_{i_\ell}$ must end, and we get that $B_{i_\ell}=B_t$.
\end{claimproof}

This completes the proof of \Cref{lem:coloring-construction}.
\end{proof}

We can now deduce the main result of this section.

\begin{theorem}\label{thm:generallower}
    Let $Q$ be a self-join-free join query with FDs $\Delta$ and $\pi$ an ordering of its variables. Let $\pi'$ be a $\Delta$-reordering of $\pi$.
    If $\width{\cn{Q}{\Delta}}(Q, \pi') > 1$, then there is no $\eps > 0$ such that for all $\delta > 0$ there is a direct access algorithm for $Q$ and $\pi$ with preprocessing time $O(|D|^{\width{\cn{Q}{\Delta}}(Q, \pi') - \varepsilon})$ and access time $O(|D|^{\delta})$, assuming the Zero-Clique Conjecture.
\end{theorem}
\begin{proof}
    According to \Cref{lem:reordering-validity}, $\pi$ and $\pi'$ define the same problem, so we will work with $\pi'$, which is consistent with $\Delta$. \Cref{lem:coloring-construction} identifies a coloring that, by \Cref{lem:star-reduction}, shows that the hardness of a star query implies the hardness of $Q$. The hardness of the star query is given by \Cref{thm:star-hard}, assuming the Zero-Clique Conjecture.
\end{proof}

\begin{example}
\Cref{thm:generallower} shows a quadratic lower bound for \Cref{ex:controller-query} since $\width{\cn{Q}{\Delta}}(Q, \pi)=2$, as seen in \Cref{ex:running-color}.
\end{example}

Note that, by \Cref{prop:reorderingsallthesame}, all $\Delta$-reorderings yield the same bound in \Cref{thm:generallower}.
Due to \Cref{prop:improvementpolymatroid}, the best algorithm we currently have is taking an (arbitrary) $\Delta$-reordering of $\pi$ and applying \Cref{thm:generalupper}.
Comparing this algorithm with the lower bound of \Cref{thm:generallower}, we get an optimal algorithm (up to sub-polynomial factors) whenever $\width{\cn{Q}{\Delta}}(Q, \pi')$ and $\width{\pb{Q}{\Delta}}(Q, \pi')$ coincide. Let us discuss the connection between these two measures next.
Of course, we cannot expect the lower bounds from Theorem~\ref{thm:generallower} to be higher than the upper bounds in Theorem~\ref{thm:generalupper}. Indeed, we can show the following:

\begin{propositionrep}\label{prop:width-comparison}
For every self-join-free join query with FDs $\Delta$ and every variable ordering~$\pi$, we have 
\begin{align*}
\width{\cn{Q}{\Delta}}(Q, \pi) \le \width{\pb{Q}{\Delta}}(Q, \pi).
\end{align*}
\end{propositionrep}
\begin{toappendix}
In the proof of Proposition~\ref{prop:width-comparison}, we will use the following result, which is implicit in~\cite[Section~6]{GottlobLVV12}. We give a simple proof for the convenience of the reader.
\begin{lemma}\label{lem:color-bounds-polymatroid}
For every self-join-free CQ $Q(X)$ and set $\Delta$ of FDs, $\cn{Q}{\Delta}(X) \le \pb{Q}{\Delta}(X)$.
\end{lemma}
\begin{proof}
Consider a valid coloring $\col$ of $Q$. Let $d:= \max_{R_i(X_i)\in \atoms(Q)} |\col(X_i)|$ be the maximum number of colors assigned to the variables of an atom by $\col$. Let $V$ be the set of variables of $Q$, then we define $f:2^V\rightarrow \mathbb R_+$ by $f(S) = \frac{|\col(S)|}{d}$.

We claim that $f$ is a polymatroid.
Clearly, $f$ is monotone. For submodularity, consider $A, B\subseteq V$. Then,
\begin{align*}
d \cdot\left(f(A\cup B) + f(A\cap B)\right) &= |\col(A\cup B)| + |\col(A\cap B)| \\&= |\col(A) \setminus \col(B)| + |\col(B) \setminus \col(A)| + |\col(A) \cap \col(B)| + |\col(A\cap B)| \\&\le |\col(A) \setminus \col(B)| + |\col(B) \setminus \col(A)| + 2 |\col(A) \cap \col(B)| \\&= |\col(A)| + |\col(B)| \\&= d\cdot \left(f(A) + f(B)\right),
\end{align*}
so $f$ is submodular and thus a polymatroid.

Clearly, $f$ is guarded by $Q$. Moreover, for every FD $Y\rightarrow z$, we have $\col(z) \subseteq \col(Y)$ because $\col$ is valid. Thus
$d\cdot f(Y\cup \{z\}) = |\col(Y\cup \{z\})| = |\col(Y)| = d\cdot f(Y)$,
so $f$ respects the FD. It follows that 
$\pb{Q}{\Delta}(X)\ge f(X) = \frac{|\col(X)|}{d} \ge \cn{Q}{\Delta}(X)$.
\end{proof}
\end{toappendix}
\begin{proof}[Proof of Proposition~\ref{prop:width-comparison}]
    The claim follows directly from Lemma~\ref{lem:color-bounds-polymatroid}.
\end{proof}

Generally, $\width{\cn{Q}{\Delta}}(Q, \pi)$ and $\width{\pb{Q}{\Delta}}(Q, \pi)$ are not equal, and the gap between them can be arbitrarily big, as we see from the following result.

\begin{propositionrep}
For all $k\in \mathbb{N}$, there is a join query $Q$, a variable order $\pi$, and a set of FDs $\Delta$, such that $\width{\pb{Q}{\Delta}}(Q, \pi')\ge \width{\cn{Q}{\Delta}}(Q, \pi') + k$ for every $\Delta$-reordering $\pi'$ of $\pi$.
\end{propositionrep}
\begin{proofsketch}
    We start from a query $Q_k'$ in~\cite{GottlobLVV12} that shows a separation between the color number and the exponent of the polymatroid bound for join queries, and we add a fresh variable and some atoms to force the variables into a common bag of the disruption-free decomposition.
\end{proofsketch}
\begin{proof}
The starting point is a join query $Q_k'(X')$ with a set of FDs $\Delta$ from~\cite[Proposition~6.11]{GottlobLVV12} such that $\cn{Q_k'}{\Delta}(X') \le 2$, and for every $N\in \mathbb N$, there is a database $D$ of size $N$ respecting $\Delta$ such that $|Q_k'(D)| \ge |D|^k$.

Treating a join query as a set of atoms, we define $Q_k := Q_k' \cup \{ R_x(x, x^*) \mid x\in X'\}$,
where~$x^*$ is a fresh variable and all $R_x$ are fresh relation symbols, and we set $X := X'\cup \{x^*\}$.
Let $\pi$ be an order of $X$, stable with respect to $\Delta$, such that $x^*$ comes last.
Note that $x^*$ does not appear in any FD, so $\pi$ can be chosen by taking any ordering of $X'$, appending~$x^*$, and performing $\Delta$-reordering.
Observe that the bag created by $x^*$ in the disruption-free decomposition is $B:= X'\cup \{x^*\}$, so $B$ contains all variables of $Q_k$.

To bound the color number $\cn{Q_k}{\Delta}(B)$, let $\col$ be an optimal coloring, so 
\begin{align*}
    \cn{Q_k}{\Delta}(B) &=
    \frac{|\col(B)|}{\max_{R_i(X_i)\in \atoms(Q_k)}|\col(X_i)|}\\
    &\le \frac{|\col(X')|}{\max_{R_i(X_i)\in \atoms(Q_k)}|\col(X_i)|} + \frac{|\col(x^*)|}{\max_{R_i(X_i)\in \atoms(Q_k)}|\col(X_i)|}\\
    &\le \frac{|\col(X')|}{\max_{R_i(X_i)\in \atoms(Q_k')}|\col(X_i)|} + \frac{|\col(x^*)|}{\max_{R_i(X_i)\in \atoms(Q_k)}|\col(X_i)|}\\
    &\le \cn{Q_k'}{\Delta}(X') + 1 \le 3,
\end{align*}
where we use that $\col$ restricted to $X'$ is a valid coloring for $Q'_k$ and $\cn{Q_k'}{\Delta}(X')\le 2$.
For every bag $B'$ of the decomposition, $B'\subseteq B$, and since $\cn{Q_k}{\Delta}$ is monotone, we get that $\cn{Q_k}{\Delta}(B')\leq \cn{Q_k}{\Delta}(B)\leq 3$. Therefore, $\width{\cn{Q}{\Delta}}(Q, \pi)\le 3$. Due to \Cref{prop:reorderingsallthesame}, for every $\Delta$-reordering $\pi'$ of $\pi$, $\width{\cn{Q}{\Delta}}(Q, \pi')=\width{\cn{Q}{\Delta}}(Q, \pi)\le 3$.

To show the lower bound on $\pb{Q_k}{\Delta}(B)$, we extend for some $N\in \mathbb N$ the database $D$ for $Q_k'$ from~\cite{GottlobLVV12} into a database for $Q_k$ by fixing $x^*$ to a constant and adding all facts in $R_x$ in which $x$ takes any possible value. Call the resulting database $D^*$, then the size of $D^*$ has increased w.r.t.~that of $D$ only by a constant factor at most. $D^*$ satisfies all constraints in $\Delta$ because so does $D$. Moreover, $|Q_k(D^*)| = |Q_k'(D)| \ge |D|^k$. Due to \Cref{thm:polymatroid}, $|Q_k(D)| \le |D|^{\pb{Q}{\Delta}(X)}$.
It follows that $\pb{Q_k}{\Delta}(B)\ge k$ and thus $\width{\pb{Q_k}{\Delta}}(Q_k, \pi) \ge \pb{Q_k}{\Delta}(B) \ge k$.
 Due to \Cref{prop:improvementpolymatroid}, for every $\Delta$-reordering~$\pi'$ of~$\pi$, $\width{\pb{Q_k}{\Delta}}(Q_k, \pi')=\width{\pb{Q_k}{\Delta}}(Q_k, \pi)\ge k$.

 Given $k\in \mathbb{N}$, we have that  
 \begin{align*}
 \width{\pb{Q_{k+3}}{\Delta}}(Q_{k+3}, \pi')\ge k+3 \ge \width{\cn{Q_{k+3}}{\Delta}}(Q_{k+3}, \pi') + k
 \end{align*}
 for every $\Delta$-reordering $\pi'$ of $\pi$.
\end{proof}

\section{Linear Preprocessing Dichotomy}\label{sec:sizepreserving}

Having shown lower and upper bounds for direct access on join queries and general FDs, we now turn to showing a dichotomy for linear preprocessing time algorithms. A crucial part of such algorithms for query answering is often getting a decomposition of the query in which all materialized bags have linear size. Since the relations corresponding to the bags of the decomposition are essentially answers to the query at hand with additional projections, the main technical contribution of this section is proving a combinatorial property of CQs under FDs that have color number~$1$. 
From this property, we can then easily compute a linear-size superset of the query answers. This superset computation forms a decomposition of linear materialized bag size that is then used similarly to the proof of \Cref{thm:generalupper} to yield a linear preprocessing time direct access algorithm for all join queries under FDs and lexicographic orders for which such an algorithm exists, assuming the Zero-Clique Conjecture.

We will now formulate and prove our linear-time dichotomy for lexicographic direct access with FDs.
To this end, we need an additional definition.
Let $Q(X)$ be a CQ with a set of FDs $\Delta$. A set of variables $S$ is \emph{$\Delta$-guarded} if there is an atom $R_i(X_i)$ of $Q$ such that $X_i \cimply S$.

\begin{theorem}\label{thm:linear-dichotomy}
    Let $Q$ be a self-join-free join query over a schema with FDs $\Delta$ and $\pi$ an ordering of its variables.
    Let $H$ be the disruption-free decomposition of $Q$ according to a $\Delta$-reordering of $\pi$.
    There is a direct access algorithm for $Q$ and $\pi$ with preprocessing time $O(|D|)$ and access time $O(\log|D|)$
    if and only, for every bag in $H$, the bag variables are $\Delta$-guarded in $Q$, assuming the Zero-Clique Conjecture.
\end{theorem}

Let us remark that, as in Theorem~\ref{thm:generallower}, the choice of $\Delta$-reordering does not matter in Theorem~\ref{thm:linear-dichotomy}. This is due to Proposition~\ref{prop:general-reordering-equal-and-better} and its application for the different width measures.
The upper bound of \Cref{thm:linear-dichotomy} follows from the following observation that $\Delta$-guardedness yields a simple algorithm that slightly improves on PANDA for queries with $\Delta$-guarded output variables with respect to the output size and computation time. 
\begin{lemma}\label{lem:computeOutputLinear}
Let $Q(X)$ be a CQ and $\Delta$ a set of FDs. If $X$ is $\Delta$-guarded in $Q$, then there is an algorithm that, given a database $D$ respecting $\Delta$, computes a relation $R_{sup}\supseteq Q(D)$ in linear time. 
\end{lemma}
\begin{proof}[Proof Sketch]
Consider an atom $R(Z)$ such that $Z \cimply X$. For every FD in $\Delta$, build a lookup table that retrieves an assignment to the FD head given an assignment to the FD body. Extend $R$ using these lookup tables to get a relation with variables $Z \cup X$ and then project to keep only $X$.
\end{proof}

The following lemma proves the combinatorial property needed for our dichotomy.
\begin{lemma}\label{lem:linearguarded}
    \label{thm:lin-size} Let $Q(X)$ be a CQ and $\Delta$ a set of FDs.
    If $\cn{Q}{\Delta}(X) = 1$, then $X$ is $\Delta$-guarded in $Q$.
\end{lemma}

We remark that $\Delta$-guardedness in fact characterizes the CQs and sets of FDs with linear output size, see~\Cref{app:linearoutput} for details.
Before proving \Cref{lem:linearguarded}, let us see how it combines with \Cref{lem:computeOutputLinear} to prove \Cref{thm:linear-dichotomy}.

\begin{proof}[Proof Sketch for \Cref{thm:linear-dichotomy}]
    For the positive direction,
    follow the proof of \Cref{thm:generalupper} on a $\Delta$-reordering, but replace the application of PANDA with \Cref{lem:computeOutputLinear}.

    For the negative direction, let $\pi'$ be a $\Delta$-reordering of $\pi$ and consider the corresponding disruption-free decomposition $H$.
    According to 
    \Cref{lem:linearguarded}, if some bag $B_i$ of $H$ is not $\Delta$-guarded, then $\cn{Q}{\Delta}(B_i)> 1$ and thus $\width{\cn{Q}{\Delta}}(Q, \pi') > 1$. 
    According to \Cref{thm:generallower} with $\varepsilon = \width{\cn{Q}{\Delta}}(Q, \pi')-1$, there exists $\delta > 0$ such that there is no direct access algorithm for $Q$ and $\pi$ with preprocessing time $O\left(|D|\right)$ and access time $O(|D|^{\delta})$, assuming the Zero-Clique Conjecture.
\end{proof}

To complete the proof of \Cref{thm:linear-dichotomy}, it remains to show \Cref{lem:linearguarded}.
We need some more definitions.
We call a set of variables $S$ \emph{compatible} with an FD $Y \rightarrow z$ if $z \notin S$ or $S \cap Y \neq \emptyset$. In other words, if $S$ contains the head of the FD, it has to contain at least one variable of the body. We say that $S$ is \emph{compatible} with a set of FDs $\Delta$ if it is compatible with every FD of $\Delta$.
Our proof is based on a linear programming characterization of the color number.

  \begin{lemmarep}
    \label{prop:lp-color}
  Let $Q(X)$ be a CQ and $\Delta$ a set of FDs. The optimal value of the following linear program on variables $\sigma_{R_i(X_i)}$ is $\cn{Q}{\Delta}(X)$.
\begin{align*}
    \text{ minimize } & \sum_{R_i(X_i) \in Q} \sigma_{R_i(X_i)} & \\
    \text{ subject to }  & \sum_{R_i(X_i) \in Q: S \cap X_i \neq \emptyset} \sigma_{R_i(X_i)} \geq 1 & \text{$\forall S\subseteq V$ compatible with $\Delta$ with $S \cap X \neq \emptyset$} \\
            & \sigma_{R_i(X_i)} \geq 0 & \text{$\forall R_i(X_i) \in Q$}
\end{align*}
\end{lemmarep}
\begin{proofsketch}
We start from a linear programming characterization of the color number by Gottlob et al.~\cite{GottlobLVV12}, take its dual, and then perform simplification steps, deleting unnecessary constraints.
\end{proofsketch}
\begin{proof}
Denote the variables of $Q$ by $V$.
A known characterization of the color number states that it is given by the optimal value of the following linear program~\cite[Proposition 6.10]{GottlobLVV12}, which we denote $P_1(Q,\Delta)$.
\begin{align*}
\text{maximize }& h(X) \\
\text{subject to }& h(X_i) \leq 1 &\text{ for every atom $R_i(X_i)$ of $Q$} \\
	& h(\{z\} \mid Y) \le 0 &\text{ for every FD $Y \rightarrow z$ of $\Delta$} \\
	& I(S \mid V \setminus S) \geq 0 &\text{ for every $S \subseteq V$}
\end{align*}
with $h(S_1 \mid S_2) = \sum_{S: S \cap S_1 \neq \emptyset, S\cap S_2 = \emptyset} I(S \mid V \setminus S)$ for every pair of disjoint sets $S_1,S_2 \subseteq V$. 
 
We denote $I_S = I(S \mid V \setminus S)$. The linear program $P_1(Q,\Delta)$ can be rewritten as the follows.
\begin{align*}
\text{maximize }& \sum_{S \cap X \neq \emptyset} I_S & \\
\text{subject to }& \sum_{S \cap X_i \neq \emptyset} I_S \leq 1 & \text{ for every atom $R_i(X_i)$ of $Q$} \\
	& \sum_{\{z\} \subseteq S \subseteq V \setminus Y} I_S \leq 0  & \text{ for every FD $Y \rightarrow z$ of $\Delta$} \\
	& I_S \geq 0 & \text{ for every $S \subseteq V$}
\end{align*}

Consider the dual of $P_1(Q,\Delta)$. It has the same optimal value, and it has one variable per constraint of $P_1(Q,\Delta)$, that is, one variable $\sigma_{R_i(X_i)}$ for each atom of $Q$ and one variable $\sigma_{Y \rightarrow z}$ for each FD of $\Delta$. 
It has one constraint $c_S$ for every variable $I_S$ of $P_1(Q,\Delta)$, and it bounds the sum of the corresponding variables by the coefficient of $I_S$ in the objective function of $P_1(Q,\Delta)$. Its objective is to minimize the sum of the variables corresponding to the constraints of $P_1(Q,\Delta)$ bounded by $1$.
The dual of $P_1(Q,\Delta)$, which we denote $P_2(Q,\Delta)$, is then
\begin{align*}
  \text{minimize } & \sum_{R_i(X_i) \in Q} \sigma_{R_i(X_i)} \\
  \text{subject to }
            & \sum_{R_i(X_i) \in Q: S \cap X_i \neq \emptyset} \sigma_{R_i(X_i)} + \sum_{Y \rightarrow z\in\Delta: \{z\} \subseteq S \subseteq V \setminus Y} \sigma_{Y \rightarrow z} \geq 1 \text{ for every $S \subseteq V$ with $S \cap X_i \neq \emptyset$} \\
            & \sum_{R_i(X_i) \in Q: S \cap X_i \neq \emptyset} \sigma_{R_i(X_i)} + \sum_{Y \rightarrow z\in\Delta: \{z\} \subseteq S \subseteq V \setminus Y} \sigma_{Y \rightarrow z} \geq 0 \text{ for every $S \subseteq V$ with $S \cap X_i=\emptyset$} \\
            & \sigma_{R_i(X_i)} \geq 0 \text{ for every $R_i(X_i) \in Q$}, \\
            & \sigma_{Y \rightarrow z} \geq 0 \text{ for every $Y \rightarrow z \in \Delta$}
\end{align*}

We now proceed to simplify $P_2(Q,\Delta)$. Since all variables of $P_2(Q,\Delta)$ are non-negative, the constraints $c_S$ with $S \cap X = \emptyset$ are trivially satisfied and can be removed.
Second, as the objective function of $P_2(Q,\Delta)$ does not use variables of the form $\sigma_{Y \rightarrow z}$, we can set all such variables to $1$ without changing the optimal value of the program. After doing so, each constraint $c_S$ for which there exists an FD $Y \rightarrow z$ in $\Delta$ such that $\{z\} \subseteq S \subseteq V \setminus Y$ is trivially satisfied.
In other words, if $S$ is not compatible with $\Delta$, then $c_S$ is trivially satisfied after setting $\sigma_{Y \rightarrow z}$ to $1$ for all FDs $Y \rightarrow z$ in $\Delta$. Hence, $P_2(Q,\Delta)$ is equivalent to the linear program stated in this lemma.
\end{proof}

We can now prove \Cref{lem:linearguarded}.

\begin{proof}[Proof of \Cref{lem:linearguarded}]
We want to prove that there is an atom $R(Z)$ of $Q$ such that $Z \cimply X$.
Denote by $P_C(Q,\Delta)$ the linear program defined in \Cref{prop:lp-color}, and denote by $c_S$ its constraint for the set $S\subseteq V$. Since the color number is $1$, the optimal value of $P_C(Q,\Delta)$ is $1$. 

We first claim that there exists an atom $R(Z)$ of $Q$ such that $\sigma_{R(Z)}$ appears in every constraint $c_S$ of $P_C(Q,\Delta)$. Indeed, consider an optimal solution $(\sigma^*)$ of $P_C(Q,\Delta)$ and a constraint $c_S$ of the program.
We have that:
\begin{align*}
  1 = \sum_{R_i(X_i) \in Q} \sigma^*_{R_i(X_i)} \geq \sum_{R_i(X_i) \in Q: S \cap X_i \neq \emptyset} \sigma^*_{R_i(X_i)} \geq 1
\end{align*}
where the first equality is since the optimal value of the program is $1$, and the last inequality is due to the constraint $c_S$.
Thus, the two sums equal $1$.
For every atom $R_i(X_i)$ of $Q$ such that $\sigma^*_{R_i(X_i)} > 0$, we have that $\sigma^*_{R_i(X_i)}$ must appear in $c_S$, because otherwise the sum above would be smaller than $1$. Hence any arbitrary atom $R(Z)$ such that $\sigma^*_{R(Z)} > 0$ has the desired property.

We now proceed to prove that $Z \cimply X$. Denote by $Z^*$ the set of variables $y$ of $Q$ such that $Z \cimply y$ and assume toward a contradiction that there exists $x \in X$ such that $x \notin Z^*$. We will construct a set $S^* \subseteq V$  such that:
\begin{itemize}
\item $c_{S^*}$ is a constraint of $P_C(Q,\Delta)$, that is, $S^*$ is compatible with $\Delta$ and $S^* \cap X \neq \emptyset$,
\item $Z \cap S^* = \emptyset$.
\end{itemize}

To do so, we construct a sequence $S^0 \subsetneq \dots \subsetneq S^r$ for some $r\in [n]$ such that $S^0 = \{x\}$ and  $Z^* \cap S^i = \emptyset$ for every $i \leq r$.
We start with $S^0 = \{x\}$ which has the desired property since $x \notin Z^*$. Now, given $S^i$, we explain how to construct $S^{i+1}$: if there is an FD $Y \rightarrow z$ in $\Delta$ such that $z \in S^i$ and $Y \cap S^i = \emptyset$, then we pick $z' \in Y \setminus Z^*$ and define $S^{i+1} = S^i \cup \{z'\}$. Observe that $z'$ always exists. Indeed, if $Y \subseteq Z^*$, then we would have $z \in Z^*$, which does not hold because $z \in S^i$ and we ensure by induction that $S^i \cap Z^* = \emptyset$.
If no such FD exists, we stop and set $S^* = S^i$. Observe that this process must stop at some point since $|S^i|$ is an increasing sequence bounded by $|V|$. 

Since we do not insert to $S^*$ elements of $Z^*$, this construction ensures that $S^{*} \cap Z^*=\emptyset$.
Now we claim that the set $S^*$ is compatible with $\Delta$. Indeed, let $Y \rightarrow z$ be an FD of $\Delta$. If $z \notin S^*$, then $S^*$ is compatible with the FD by definition. If $z \in S^*$, then we must have $S^* \cap Y \neq \emptyset$ by construction of $S^*$, and so $S^*$ is compatible with the FD. Overall, $S^*$ is compatible with $\Delta$, and $X \cap S^* \neq \emptyset$ since $x \in S^*$.

The linear program $P_C(Q,\Delta)$ contains the constraint $c_{S^*}$, but since $Z \cap S^* = \emptyset$, $\sigma_{R(Z)}$ does not appear in $c_{S^*}$, which contradicts the established property of $R(Z)$. Hence, we have that $x \in Z^*$ for every $x \in X$, and so $Z \cimply X$.
\end{proof}

\begin{toappendix}
    
\section{Characterizing Linear Output Size}\label{app:linearoutput}

In this Appendix, we give new characterizations of the CQs with guaranteed linear output size.
We say that $Q$ has \emph{linear output size} under $\Delta$ if, for every database $D$ respecting $\Delta$, we have that $|Q(D)|=O(|D|)$.
As a stricter property, we say that $Q$ is \emph{size preserving} under $\Delta$ if, over such databases, we have that $|Q(D)| \le |D|$.
We previously only defined $\cn{Q}{\Delta}$ and $\pb{Q}{\Delta}$ when $Q$ is a join query. In the following, we also use this notation when $Q$ is a CQ by ignoring the projections.
The main result of this section is the following Theorem.
\begin{theorem}\label{conj:linear}
    Let $Q(X)$ be a self-join-free CQ, and let $\Delta$ be a set of FDs. Then, the following statements are equivalent:
    \begin{enumerate}
        \item $Q(X)$ is size preserving under $\Delta$.\label{item:SP}
        \item $Q(X)$ has linear output size under $\Delta$.\label{item:lin}
        \item $\cn{Q}{\Delta}(X) = 1$.\label{item:CN}
        \item $\pb{Q}{\Delta}(X) = 1$.\label{item:PB}
        \item $X$ is $\Delta$-guarded in $Q$.\label{item:FG}
    \end{enumerate}
\end{theorem}

Properties \ref{item:SP}, \ref{item:lin} and \ref{item:CN} are already known to be equivalent~\cite[Theorem 6.1]{GottlobLVV12}. It remains to prove the equivalence of Properties \ref{item:CN}, \ref{item:PB}, and \ref{item:FG}.

We remark that there are some cases in which \Cref{conj:linear} is already essentially known:
In case there are no FDs, the color number, the fractional edge cover number, and the polymatroid bound coincide~\cite[Proposition~3.2]{Khamis0S16}~\cite[Section~3.1]{GottlobLVV12}. In addition, in this case, the fractional edge cover number of a variable set $S$ is $1$ if and only if there is a single atom containing $S$ (which is equivalent to $S$ being $\Delta$-guarded in that case).
In case all FDs are unary, some of these equivalences are known and implicit in prior work~\cite{GottlobLVV12,abo2016computing,Suciu23}. 

\begin{lemma}\label{lem:FDguarded-to-polymatroid}
    Let $Q(X)$ be a CQ, and let $\Delta$ be a set of FDs. 
    If $X$ is $\Delta$-guarded in $Q$, then $\pb{Q}{\Delta}(X)=1$.
\end{lemma}
\begin{proof}
    $\pb{Q}{\Delta}(X)\ge 1$ since the function that assigns $1$ to every set is a polymatroid. So it suffices to show the upper bound. If $X$ is $\Delta$-guarded, there is an atom $R_i(X_i)$ with $X_i \cimply X$. For any polymatroid respecting $\Delta$, we then get $h(X)\le h(X X_i) = h(X_i) \le 1$, where the first inequality is by monotonicity, the second is due to the FDs, and the third is since $h$ is guarded by $Q$. Thus, $\pb{Q}{\Delta}(X)\le 1$.
  \end{proof}

\begin{proof}[Proof of \Cref{conj:linear}]
Properties \ref{item:SP}, \ref{item:lin} and \ref{item:CN} are known to be equivalent~\cite[Theorem 6.1]{GottlobLVV12}.
We prove next the equivalence of Properties \ref{item:CN}, \ref{item:PB}, and \ref{item:FG}.
Observe that $\cn{Q}{\Delta}(X)\ge 1$ as coloring all variables with the same single color is always a valid coloring.
If $\pb{Q}{\Delta}(X)= 1$, then using \Cref{lem:color-bounds-polymatroid}, $1\le \cn{Q}{\Delta}(X)\le \pb{Q}{\Delta}(X)= 1$, and so $\cn{Q}{\Delta}(X)= 1$.
By \Cref{lem:linearguarded}, if $\cn{Q}{\Delta}(X)= 1$, then $X$ is $\Delta$-guarded.
Finally, by \Cref{lem:FDguarded-to-polymatroid}, if $X$ is $\Delta$-guarded, then $\pb{Q}{\Delta}(X)= 1$.
\end{proof}

\end{toappendix}

\section{Connection between the approaches}\label{sec:connection}

In this section, we compare the extension-based approach from Section~\ref{sec:reorder} with that based on disruption-free decompositions from Section~\ref{sct:algorithm}.
We have already seen in \Cref{ex:running-polymatroid} that the approach from Section~\ref{sct:algorithm} is sometimes more efficient than the other.
We show that it is always at least as good as that from \Cref{sec:reorder}, and that in the case of unary FDs, the two approaches yield the same complexity.

For the rest of this section, consider a self-join-free join query $Q$, a set of FDs $\Delta$, and an ordering $\pi$ of the variables of $Q$ that is consistent with respect to $\Delta$. Deonte by $Q^+$ be the extension of $Q$, and remark that $\pi^+=\pi$ since it is consistent.
Recall that $\width{\fec{Q^+}}(Q^+,\pi)$ is the incompatibility number of $Q^+$ and $\pi$, and it determines the runtime exponent for the approach from \Cref{sec:reorder} as given in \Cref{cor:ext-alg}. 
The polymatroid-based measure $\width{\pb{Q}{\Delta}}(Q, \pi)$ determines the runtime exponent for the approach from \Cref{sct:algorithm} as given in \Cref{thm:generalupper}.
This section compares these two measures.

In this section, we use the following notation. We denote $\pi=(v_1, \ldots, v_n)$. We let $b_i$ denote the bag created by $v_i$ in the disruption-free decomposition of $Q$, and let $B_i$ be the bag created by $v_i$ in that of the extension $Q^+$.
Given a set $S$ of variables, we define the set $S^+$ as the fixpoint of the following extension step, starting with $S^+ := S$: if there is an FD $Y \rightarrow z$ such that $Y \subseteq S$ and $z \notin S^+$, then add $z$ to $S^+$. Clearly, the choice of the order in which we handle the FDs and add variables in the construction of $S^+$ does not change the end result, so $S^+$ is well-defined.

\subsection{Superiority of the polymatroid-based approach}

The following proposition shows that the approach from \Cref{sct:algorithm} is always at least as good as the extension-based approach.

\begin{proposition}\label{prop:comparison}
	Let $Q$ be a self-join-free join query, $\Delta$ a set of FDs, and $\pi$ a consistent ordering with respect to $\Delta$. Moreover, let $Q^+$ be the extension of $Q$. Then,
	\begin{align*}
		\width{\pb{Q}{\Delta}}(Q, \pi) \le \width{\fec{Q^+}}(Q^+,\pi).
	\end{align*}
\end{proposition}

Remember that in Example~\ref{ex:controller-query} and Example~\ref{ex:running-polymatroid}, we have seen that the two width measures $\width{\pb{Q}{\Delta}}(Q, \pi)$ and $\width{\fec{Q^+}}(Q^+,\pi)$ differ in some cases, so the inequality of Proposition~\ref{prop:comparison} can in general be strict.

We will show Proposition~\ref{prop:comparison} in two steps. We start with the following.

\begin{lemma}\label{lem:comparison1}
	Let $Q$, $Q^+$, $\Delta$, and $\pi$ be as in Proposition~\ref{prop:comparison}. Then, we have that $\width{\pb{Q}{\Delta}}(Q, \pi) \le \width{\pb{Q^+}{\emptyset}}(Q^+,\pi)$.
\end{lemma}
\begin{proof}
	We show that $\pb{Q}{\Delta}(b_i) \le \pb{Q^+}{\emptyset}(B_i)$. By definition, we have that $\pb{Q}{\Delta}(b_i)$ is the solution of 
	\begin{align*}
		\text{maximize }& h(b_i) \\
		\text{subject to }& h(X_i) \leq 1 &\text{ for every atom $R_i(X_i)$ of $Q$} \\
		& h(\{z\} \mid Y) \le 0 &\text{ for every FD $Y \rightarrow z$ of $\Delta$} \\
		& \text{$h$ is a polymatroid.}
	\end{align*}
	Let $h^*$ be the polymatroid that maximizes the value $h(b_i)$ in this system.
	As defined at the beginning of this section, given a set $V$ of variables, we denote by $V^+$ the set obtained by extending $V$ according to $\Delta$. Then we have, due to the constraints for the FD, that $h^*(V^+) \le h^*(V)$ and thus, for every atom $R_i(X_i)$, we have $h^*(X_i^+) \le h^*(X_i)$. It follows that $h^*$ is a feasible solution for 
	\begin{align*}
		\text{maximize }& h(B_i) \\
		\text{subject to }& h(X_i^+) \leq 1 &\text{ for every atom $R_i(X_i^+)$ of $Q^+$} \\
		& \text{$h$ is a polymatroid.}
	\end{align*}
	The optimal solution of this linear program is $\pb{Q^+}{\emptyset}(B_i)$ by definition. Due to the monotonicity of the polymatroid $h^*$ an since $b_i\subseteq B_i$, we have that $h^*(b_i)\le h^*(B_i)$. Overall, we get that $\pb{Q}{\Delta}(b_i)=h^*(b_i)\le h^*(B_i) \le \pb{Q^+}{\emptyset}(B_i)$.
	
	The claim follows since $\width{\pb{Q}{\Delta}}(Q, \pi) = \max_{i\in n} \pb{Q}{\Delta}(b_i) \le \max_{i\in n} \pb{Q^+}{\emptyset}(B_i) = \width{\pb{Q^+}{\emptyset}}(Q^+,\pi)$.
\end{proof}

It remains to show the following.

\begin{lemma}\label{lem:comparison2}
	Let $Q^+$ be a self-join-free join query and let $\pi$ be a variable order of $Q^+$. Then, 
	$\width{\pb{Q^+}{\emptyset}}(Q, \pi) = \width{\fec{Q^+}}(Q^+,\pi)$.
\end{lemma}
\begin{proof}
	It suffices to show that, for every variable set of the disruption-free decomposition of $Q^+$ and $\pi$, we have $\pb{Q^+}{\emptyset}(B) = \fec{Q^+}(B)$. However, this is exactly the so-called Modularization Lemma, Lemma~3.1 in~\cite{Khamis0S16}.
\end{proof}

Combining Lemma~\ref{lem:comparison1} and Lemma~\ref{lem:comparison2} directly yields Proposition~\ref{prop:comparison}.

\subsection{Equivalence in case of unary FDs}

We show next that, for unary FDs, the two approaches give the same width values and thus the same runtimes.
We remark that this can be concluded immediately, conditioned on the Zero-Clique Conjecture, from \Cref{prop:comparison} and \Cref{cor:unary-inc} since the extension approach is optimal in the unary case and the polymatroid bound approach is at least as good. In this section, we show that this also holds unconditionally.

\begin{toappendix}
	We will need the fact that every bag of the extended query is contained in some extension of a bag of the original query. More specifically, we prove the following lemma.
	\begin{lemma}\label{lem:bags-contained-no-preced}
		Let $j$ be the smallest index such that $v_j$ implies $v_i$. Then, $B_i\subseteq b_j^+$.
	\end{lemma}
	Remark that, in \Cref{lem:bags-contained-no-preced}, $j\le i$ because $v_i$ implies itself regardless of the set of FDs.
	We proceed to showing claims that will help prove the lemma.
	
	\begin{claim}\label{claim:implies-to-path}
		If $v_a\cimply v_b$ (with unary FDs), then there is a path between $v_a$ and $v_b$ in $Q$ that uses only variables that imply $v_b$.
	\end{claim}
	\begin{proof}
		Since $v_a\cimply v_b$, there is a sequence $v_a=v_1,\ldots,v_k=v_b$ with an FD $v_i\rightarrow v_{i+1}$ in $\Delta$ for every $i\in[k-1]$. This sequence forms the required path.
	\end{proof}
	
	\begin{claim}\label{claim:no-imply-after}
		Let $j$ be the smallest such that $v_j$ implies $v_i$.
		If $v_a\cimply v_b$, with $b\geq i$, then $a\geq j$.
	\end{claim}
	\begin{proof}
		Assume by way of contradiction that $a<j$. Since $a<j\leq i\leq b$ and it is a consistent ordering, $v_a\cimply v_b$ means that $v_i$ must be implied by $v_a$ or a variable preceding it, which is a contradiction to the minimality of $j$.
	\end{proof}
	
	\begin{claim}\label{claim:no-imply-path}
		Let $j$ be the smallest such that $v_j$ implies $v_i$.
		If $v_a,v_b$ are neighbors in $Q^+$ with $a,b\geq i$, then there is a path between $v_a$ and $v_b$ in $Q$ that uses only variables $\geq j$.
	\end{claim}
	\begin{proof}
		If $v_a,v_b$ are neighbors in $Q$, the claim is immediate. Otherwise, since they are neighbors in $Q^+$, there exist variables $v_c$ and $v_d$, neighbors in $Q$, that imply $v_a$ and $v_b$ respectively (it could be the case that $c=d$).
		By applying \Cref{claim:implies-to-path} twice, there is a path $v_a,\ldots,v_c,v_d\ldots,v_b$ in $Q$ using only variables that imply $v_a$ or $v_b$. By using \Cref{claim:no-imply-after} on all variables of the path, and since $a,b\geq i$, these variables are all $\geq j$.
	\end{proof}
	
	\begin{proof}[Proof of \Cref{lem:bags-contained-no-preced}]
		Let $v_t\in B_i$. We need to show that $v_t\in b_j^+$.
		We have that $v_j\in b_j\subseteq b_j^+$, and since $v_j\cimply v_i$, we have that $v_i\in b_j^+$.
		Thus, if $t\in\{i,j\}$, we have that $v_t\in b_j^+$.
		Recall that $j\le i$ because $v_i$ implies itself.
		Thus, it remains to handle the case that $j<t<i$ and the case that $t<j$.
		
		Consider the case that $j<t<i$.
		First we claim that $v_j$ is not implied by any variable preceding it. Indeed, since $v_j\cimply v_i$, such a variable would transitively imply $v_i$, contradicting the minimality of $j$.
		Since the ordering is consistent, no variable preceding $v_j$ implies it, and $v_j\cimply v_i$, we get that $v_j$ implies all variables between $v_j$ and $v_i$. Since $j<t<i$, we have that $v_j\cimply v_t$. Since $v_j\in b_j$, we conclude that $v_t\in b_j ^+$.
		
		The last case is that $t<j\le i$.
		By \Cref{claim:implies-to-path}, since $v_j\cimply v_i$, there is a path between $v_j$ and $v_i$ that only goes through variables that imply $v_i$. Due to the minimality of $j$, these variables are all $\geq j$.
		Since $v_y\in B_i$, there is a path $v_t,v_k,\ldots,v_i$ in $Q^+$ that only goes through variables $\geq i$.
		By applying \Cref{claim:no-imply-path} on the edges of $v_k,\ldots,v_i$, we get a path from $v_k$ to $v_i$ in $Q$ that only goes through variables $\geq j$. By combining these paths, we get a path from $v_k$ to $v_j$ that only goes through variables $\geq j$.
		It follows that $v_k \in S_j$, where $S_j$ is the set from the definition of the disruption-free decomposition of $Q$.
		Thus, if $v_t$ neighbors $v_k$ in $Q$, we conclude that $v_t\in b_j\subseteq b_j^+$, and we are done.
		Otherwise, since $v_t$ neighbors $v_k$ in $Q^+$, there exist variables $v_a$ and $v_b$, neighbors in $Q$, that imply $v_t$ and $v_k$ respectively (again, it could be that $a=b$).
		Since $v_b\cimply v_k$, from \Cref{claim:implies-to-path} and \Cref{claim:no-imply-after}, we get a path between $v_b$ and $v_k$ using only variables $\geq j$.
		Since $v_a\cimply v_t$, from \Cref{claim:implies-to-path} we get a path from $v_a$ to $v_t$ on which all variables imply $v_t$. Let $v_m$ be the first variable on this path smaller than $j$ (it exists because $t<j$). By composing the paths we discovered, we get a path $v_m,\ldots,v_a,v_b\ldots,v_k\ldots,v_j$ in $Q$ that, other than $v_m$, only uses variables $\geq j$. Thus, $v_m\in b_j$, and since $v_m\cimply v_t$, we get that $v_t\in b_j^+$.
	\end{proof}
	
	We use Lemma~\ref{lem:bags-contained-no-preced} to relate $\width{\fec{Q^+}}(Q^+,\pi)$, $\width{\pb{Q^+}{\emptyset}}(Q, \pi)$ and $\width{\pb{Q^+}{\emptyset}}(Q, \pi)$.
\end{toappendix}

\begin{propositionrep}\label{prop:unary-comparison}
	Let $Q$ be a self-join-free join query, $\Delta$ a set of FDs, and $\pi$ a consistent ordering $\pi$ with respect to $\Delta$. Let moreover $Q^+$ be the extension of $Q$. Then,
	\begin{align*}
		\width{\cn{Q}{\Delta}}(Q, \pi) = \width{\pb{Q}{\Delta}}(Q, \pi) = \width{\fec{Q^+}}(Q^+,\pi).
	\end{align*}
\end{propositionrep}
\begin{appendixproof}
	By \Cref{prop:width-comparison} and Proposition~\ref{prop:comparison}, we have 
	\begin{align*}
		\width{\cn{Q}{\Delta}}(Q, \pi) \le \width{\pb{Q}{\Delta}}(Q, \pi) \le \width{\fec{Q^+}}(Q^+,\pi), 
	\end{align*}
	so it suffices to show $\width{\fec{Q^+}}(Q^+,\pi) \le \width{\cn{Q}{\Delta}}(Q, \pi)$.
	
	Let $\col$ be  any coloring of $\{v_1, \ldots, v_n\}$. We define a new coloring $\col'$ by setting for every variable $v$
	\begin{align*}
		\col'(v) := \col(v) \cup \bigcup_{v\cimply z \in \Delta}\col(z).
	\end{align*}
	Clearly, $\col'$ is $\Delta$-valid by construction.
	
	We claim that for every set $S$, we have $\col'(S) = \col(S^+)$. For every variable $z\in S^+$, there is a variable $y\in S$ such that $y\cimply z$. By definition of $\col'$, we have that $\col(z)\subseteq \col'(y)$. Consequently, $\col(S^+)\subseteq \col'(S)$. Finally, if $c\in\col'(v)$ for some $v\in S$, then either $c\in\col(v)$ or $c\in\col(z)$ for some $v\cimply z \in \Delta$. In any case, $c \in\col(S^+)$.
	
	Consider any bag $B_i$ of the disruption-free decomposition of $Q^+$. By Lemma~\ref{lem:bags-contained-no-preced}, there is an index $j$ such that $B_i \subseteq b_j^+$. It follows that
	\begin{align*}
		\cn{Q^+}{\emptyset}(B_i) &= \max_{\col_1} \frac{|{\col_1}(B_i)|}{\max_{R_i(X_i)\in \atoms(Q)}|{\col_1}(X_i^+)|}\\
		&\le \max_{\col_1} \frac{|{\col_1}(b_j^+)|}{\max_{R_i(X_i)\in \atoms(Q)}|{\col_1}(X_i^+)|}\\
		&= \max_{\col_1} \frac{|{\col_1}'(b_j)|}{\max_{R_i(X_i)\in \atoms(Q)}|{\col_1}'(X_i)|}\\
		&= \max_{\text{$\Delta$-valid } \col_2} \frac{|\col_2(b_j)|}{\max_{R_i(X_i)\in \atoms(Q)}|\col_2(X_i)|}\\
		&= \cn{Q}{\Delta}(b_j),
	\end{align*}
	By Section~3.1 of~\cite{GottlobLVV12}, we have for all sets $S$ that $\fec{Q^+}(S) = \cn{Q^+}{\emptyset}(S)$, so we get $\fec{Q^+}(B_i) \le \cn{Q}{\Delta}(b_j)$. It follows that $\width{\fec{Q^+}}(Q^+,\pi) \le \width{\cn{Q}{\Delta}}(Q, \pi)$, which completes the proof.
\end{appendixproof}

\section{Conclusion}\label{sec:conclusion}

In this paper, we proved upper and lower bounds for lexicographic direct access for self-join-free join queries that take FDs into account. Unfortunately, these bounds are not tight in general, which is likely related to the fact that there is no known worst-case optimal algorithm for join queries with functional dependencies. While the PANDA algorithm reaches the polymatroid bound, this bound is in general bigger than the so-called entropic bound which determines the size of query results asymptotically~\cite{Khamis0S16}, and thus PANDA is not worst-case optimal when taking FDs into consideration. Since we use a join algorithm as a black-box, improving upon PANDA, if this is even possible, might be necessary to get tight bounds in our setting.   
That said, even now our results allow us to characterize the combinations of a query, a variable order, and a set of FDs that admit logarithmic access time after linear preprocessing. We also proved a tight characterization of the preprocessing time required for logarithmic access time in case all FDs are unary.

We studied two approaches in this paper, and saw using our running example (\Cref{ex:controller-query}) that the approach of taking the FDs into account while inspecting the bags of a disruption-free decomposition (\Cref{{thm:generalupper}} applied to a $\Delta$-reordering) sometimes yields better complexity than the simple approach of using a reordered extension (\Cref{{cor:ext-alg}}).
We also saw that, in terms of complexity, this approach is never worse, so we can always use the approach from \Cref{sec:algorithm}. However, if all FDs are unary, the two approaches yield the same complexity, so it could make sense to use the simpler approach from \Cref{sec:reorder}.

A natural next step is to try and generalize our results.
On the query side, one could check whether the direct-access lower-bound techniques for queries with self-joins~\cite{bringmann2024stars} can be combined with the generalization of the color number to such queries~\cite{GottlobLVV12}.
On the constraints side, as PANDA also supports degree constraints~\cite{Khamis0S16}, our algorithmic approach can remain the same to support such constraints, and it would be interesting to try and prove a matching lower bound.
Finally, we would like to see whether the techniques from this paper can be used for other query answering tasks, such as enumeration, counting, or more general aggregation.

\bibliography{references}

\begin{thebibliography}{10}

\bibitem{AbiteboulHV95}
Serge Abiteboul, Richard Hull, and Victor Vianu.
\newblock {\em Foundations of Databases}.
\newblock Addison-Wesley, 1995.
\newblock URL: \url{http://webdam.inria.fr/Alice/}.

\bibitem{KhamisHS25}
Mahmoud {Abo Khamis}, Xiao Hu, and Dan Suciu.
\newblock Fast matrix multiplication meets the submodular width.
\newblock {\em Proc. {ACM} Manag. Data}, 3(2):98:1--98:26, 2025.
\newblock \href {https://doi.org/10.1145/3725235} {\path{doi:10.1145/3725235}}.

\bibitem{abo2016computing}
Mahmoud {Abo Khamis}, Hung~Q Ngo, and Dan Suciu.
\newblock Computing join queries with functional dependencies.
\newblock In {\em Proceedings of the 35th ACM SIGMOD-SIGACT-SIGAI Symposium on
  Principles of Database Systems}, pages 327--342, 2016.

\bibitem{Khamis0S16}
Mahmoud {Abo Khamis}, Hung~Q. Ngo, and Dan Suciu.
\newblock What do shannon-type inequalities, submodular width, and disjunctive
  datalog have to do with one another?
\newblock {\em CoRR}, abs/1612.02503, 2016.
\newblock URL: \url{http://arxiv.org/abs/1612.02503}, \href
  {https://arxiv.org/abs/1612.02503} {\path{arXiv:1612.02503}}.

\bibitem{KhamisNS25}
Mahmoud {Abo Khamis}, Hung~Q. Ngo, and Dan Suciu.
\newblock {PANDA:} query evaluation in submodular width.
\newblock {\em TheoretiCS}, 4, 2025.
\newblock URL: \url{https://doi.org/10.46298/theoretics.25.12}, \href
  {https://doi.org/10.46298/THEORETICS.25.12}
  {\path{doi:10.46298/THEORETICS.25.12}}.

\bibitem{AboKhamisNS25}
Mahmoud {Abo Khamis}, Hung~Q. Ngo, and Dan Suciu.
\newblock Pandaexpress: a simpler and faster {PANDA} algorithm.
\newblock {\em CoRR}, abs/2512.10217, 2025.
\newblock URL: \url{https://doi.org/10.48550/arXiv.2512.10217}, \href
  {https://arxiv.org/abs/2512.10217} {\path{arXiv:2512.10217}}, \href
  {https://doi.org/10.48550/ARXIV.2512.10217}
  {\path{doi:10.48550/ARXIV.2512.10217}}.

\bibitem{Adler06}
Isolde Adler.
\newblock {\em Width functions for hypertree decompositions (Weitefunktionen
  f{\"{u}}r Hyperbaumzerlegungen)}.
\newblock PhD thesis, University of Freiburg, Germany, 2006.
\newblock URL: \url{https://freidok.uni-freiburg.de/data/2468}.

\bibitem{BaganDGO08}
Guillaume Bagan, Arnaud Durand, Etienne Grandjean, and Fr{\'{e}}d{\'{e}}ric
  Olive.
\newblock Computing the jth solution of a first-order query.
\newblock {\em {RAIRO} Theor. Informatics Appl.}, 42(1):147--164, 2008.
\newblock \href {https://doi.org/10.1051/ita:2007046}
  {\path{doi:10.1051/ita:2007046}}.

\bibitem{BerkholzGS20}
Christoph Berkholz, Fabian Gerhardt, and Nicole Schweikardt.
\newblock Constant delay enumeration for conjunctive queries: a tutorial.
\newblock {\em {ACM} {SIGLOG} News}, 7(1):4--33, 2020.
\newblock \href {https://doi.org/10.1145/3385634.3385636}
  {\path{doi:10.1145/3385634.3385636}}.

\bibitem{bb:thesis}
Johann Brault-Baron.
\newblock {\em De la pertinence de l’{\'e}num{\'e}ration: complexit{\'e} en
  logiques propositionnelle et du premier ordre}.
\newblock PhD thesis, Universit{\'e} de Caen, 2013.

\bibitem{bringmann2024stars}
Karl Bringmann, Nofar Carmeli, and Stefan Mengel.
\newblock Tight fine-grained bounds for direct access on join queries.
\newblock {\em {ACM} Trans. Database Syst.}, 50(1):1:1--1:44, 2025.
\newblock \href {https://doi.org/10.1145/3707448} {\path{doi:10.1145/3707448}}.

\bibitem{CapelliIS25}
Florent Capelli, Oliver Irwin, and Sylvain Salvati.
\newblock A simple algorithm for worst case optimal join and sampling.
\newblock In Sudeepa Roy and Ahmet Kara, editors, {\em 28th International
  Conference on Database Theory, {ICDT} 2025, Barcelona, Spain, March 25-28,
  2025}, volume 328 of {\em LIPIcs}, pages 23:1--23:19. Schloss Dagstuhl -
  Leibniz-Zentrum f{\"{u}}r Informatik, 2025.
\newblock URL: \url{https://doi.org/10.4230/LIPIcs.ICDT.2025.23}, \href
  {https://doi.org/10.4230/LIPICS.ICDT.2025.23}
  {\path{doi:10.4230/LIPICS.ICDT.2025.23}}.

\bibitem{carmeli2020enumeration}
Nofar Carmeli and Markus Kr{\"o}ll.
\newblock Enumeration complexity of conjunctive queries with functional
  dependencies.
\newblock {\em Theory of Computing Systems}, 64(5):828--860, 2020.

\bibitem{carmeli2023linearp}
Nofar Carmeli, Nikolaos Tziavelis, Wolfgang Gatterbauer, Benny Kimelfeld, and
  Mirek Riedewald.
\newblock Tractable orders for direct access to ranked answers of conjunctive
  queries.
\newblock {\em ACM Transactions on Database Systems}, 48(1):1--45, 2023.

\bibitem{carmeli2022random}
Nofar Carmeli, Shai Zeevi, Christoph Berkholz, Alessio Conte, Benny Kimelfeld,
  and Nicole Schweikardt.
\newblock Answering (unions of) conjunctive queries using random access and
  random-order enumeration.
\newblock {\em ACM Transactions on Database Systems (TODS)}, 47(3):1--49, 2022.

\bibitem{DeedsM25}
Kyle Deeds and Timo~Camillo Merkl.
\newblock Partition constraints for conjunctive queries: Bounds and worst-case
  optimal joins.
\newblock In Sudeepa Roy and Ahmet Kara, editors, {\em 28th International
  Conference on Database Theory, {ICDT} 2025, Barcelona, Spain, March 25-28,
  2025}, volume 328 of {\em LIPIcs}, pages 17:1--17:18. Schloss Dagstuhl -
  Leibniz-Zentrum f{\"{u}}r Informatik, 2025.
\newblock URL: \url{https://doi.org/10.4230/LIPIcs.ICDT.2025.17}, \href
  {https://doi.org/10.4230/LIPICS.ICDT.2025.17}
  {\path{doi:10.4230/LIPICS.ICDT.2025.17}}.

\bibitem{DeepHK20}
Shaleen Deep, Xiao Hu, and Paraschos Koutris.
\newblock Fast join project query evaluation using matrix multiplication.
\newblock In David Maier, Rachel Pottinger, AnHai Doan, Wang{-}Chiew Tan,
  Abdussalam Alawini, and Hung~Q. Ngo, editors, {\em Proceedings of the 2020
  International Conference on Management of Data, {SIGMOD} Conference 2020,
  online conference [Portland, OR, USA], June 14-19, 2020}, pages 1213--1223.
  {ACM}, 2020.
\newblock \href {https://doi.org/10.1145/3318464.3380607}
  {\path{doi:10.1145/3318464.3380607}}.

\bibitem{GogaczT17}
Tomasz Gogacz and Szymon Torunczyk.
\newblock Entropy bounds for conjunctive queries with functional dependencies.
\newblock In Michael Benedikt and Giorgio Orsi, editors, {\em 20th
  International Conference on Database Theory, {ICDT} 2017, Venice, Italy,
  March 21-24, 2017}, volume~68 of {\em LIPIcs}, pages 15:1--15:17. Schloss
  Dagstuhl - Leibniz-Zentrum f{\"{u}}r Informatik, 2017.
\newblock URL: \url{https://doi.org/10.4230/LIPIcs.ICDT.2017.15}, \href
  {https://doi.org/10.4230/LIPICS.ICDT.2017.15}
  {\path{doi:10.4230/LIPICS.ICDT.2017.15}}.

\bibitem{GottlobLVV12}
Georg Gottlob, Stephanie~Tien Lee, Gregory Valiant, and Paul Valiant.
\newblock Size and treewidth bounds for conjunctive queries.
\newblock {\em J. {ACM}}, 59(3):16:1--16:35, 2012.
\newblock \href {https://doi.org/10.1145/2220357.2220363}
  {\path{doi:10.1145/2220357.2220363}}.

\bibitem{Hu25}
Xiao Hu.
\newblock Output-optimal algorithms for join-aggregate queries.
\newblock {\em Proc. {ACM} Manag. Data}, 3(2):104:1--104:27, 2025.
\newblock \href {https://doi.org/10.1145/3725241} {\path{doi:10.1145/3725241}}.

\bibitem{Kent83}
William Kent.
\newblock A simple guide to five normal forms in relational database theory.
\newblock {\em Communications of the ACM}, 26(2):120--125, 1983.

\bibitem{LokshtanovMS11}
Daniel Lokshtanov, D{\'{a}}niel Marx, and Saket Saurabh.
\newblock Lower bounds based on the exponential time hypothesis.
\newblock {\em Bull. {EATCS}}, 105:41--72, 2011.
\newblock URL: \url{http://eatcs.org/beatcs/index.php/beatcs/article/view/92}.

\bibitem{Mengel25}
Stefan Mengel.
\newblock Lower bounds for conjunctive query evaluation.
\newblock In Floris Geerts and Benny Kimelfeld, editors, {\em Companion of the
  44th Symposium on Principles of Database Systems, {PODS} 2025, Berlin,
  Germany, June 22-27, 2025}, page~5. {ACM}, 2025.
\newblock \href {https://doi.org/10.1145/3722234.3725824}
  {\path{doi:10.1145/3722234.3725824}}.

\bibitem{Ngo2018}
Hung~Q. Ngo.
\newblock Worst-case optimal join algorithms: Techniques, results, and open
  problems.
\newblock In Jan~Van den Bussche and Marcelo Arenas, editors, {\em Proceedings
  of the 37th {ACM} {SIGMOD-SIGACT-SIGAI} Symposium on Principles of Database
  Systems, Houston, TX, USA, June 10-15, 2018}, pages 111--124. {ACM}, 2018.
\newblock \href {https://doi.org/10.1145/3196959.3196990}
  {\path{doi:10.1145/3196959.3196990}}.

\bibitem{NgoPRR18}
Hung~Q. Ngo, Ely Porat, Christopher R{\'{e}}, and Atri Rudra.
\newblock Worst-case optimal join algorithms.
\newblock {\em J. {ACM}}, 65(3):16:1--16:40, 2018.
\newblock \href {https://doi.org/10.1145/3180143} {\path{doi:10.1145/3180143}}.

\bibitem{Veldhuizen14}
Todd~L. Veldhuizen.
\newblock Triejoin: {A} simple, worst-case optimal join algorithm.
\newblock In Nicole Schweikardt, Vassilis Christophides, and Vincent Leroy,
  editors, {\em Proc. 17th International Conference on Database Theory (ICDT),
  Athens, Greece, March 24-28, 2014}, pages 96--106. OpenProceedings.org, 2014.
\newblock URL: \url{https://doi.org/10.5441/002/icdt.2014.13}, \href
  {https://doi.org/10.5441/002/ICDT.2014.13}
  {\path{doi:10.5441/002/ICDT.2014.13}}.

\end{thebibliography}

\end{document}